\documentclass{amsart}
\usepackage[a4paper, total={6in, 8in}]{geometry}
\usepackage{graphicx}
\usepackage{subfig}
\usepackage{amssymb,amscd}
\usepackage{color}
\usepackage{enumitem}
\usepackage{algorithm2e}
\usepackage{algorithmic}
\usepackage{hyperref}
\RestyleAlgo{ruled}
\SetKwComment{Comment}{/* }{ */}
\LinesNumbered
\usepackage{bm}

 \newtheorem{thm}{Theorem}[section]
 
 \newtheorem{cor}[thm]{Corollary}
 \newtheorem{lemma}[thm]{Lemma}
 \newtheorem{prop}[thm]{Proposition}
 \theoremstyle{definition}
 \newtheorem{defn}[thm]{Definition}
 \newtheorem{exmp}[thm]{Example}

 \theoremstyle{remark}
 \newtheorem{rem}[thm]{Remark}
 \newtheorem{claim}[thm]{Claim}

 \numberwithin{equation}{subsection}

\newcommand{\cS}{\mathcal{S}}

\newcommand{\cC}{\text{$\mathcal{C}$}} 

\newcommand{\cP}{\text{$\mathcal{P}$}}

\newcommand{\Lvl}{\operatorname{Lvl}}
\newcommand{\diam}{\operatorname{diam}}
\newcommand{\vol}{\operatorname{vol}}
\newcommand{\str}{\operatorname{Stretch}}

        \newcommand{\field}[1]{\text{$\mathbb{#1}$}}
        \newcommand{\N}{\field{N}}
        \newcommand{\Z}{\field{Z}}
        
        \newcommand{\R}{\field{R}}

\newdimen\theight
\def\TeXref#1{%
             \leavevmode\vadjust{\setbox0=\hbox{{\tt
                     \quad\quad  {\small \textrm #1}}}%
             \theight=\ht0
             \advance\theight by \lineskip
             \kern -\theight \vbox to
             \theight{\rightline{\rlap{\box0}}%
             \vss}%
             }}%

\begin{document}

\title{Minimal $L^p$-congestion spanning trees on weighted graphs}

\author{Alberto Castej\'on Lafuente}
\email{acaste@uvigo.gal}

\author{Emilio Est\'evez}
\email{emestevez@uvigo.gal}

\author{Carlos Meni\~no Cot\'on$^\dagger$}
\email{carlos.menino@uvigo.es}
\thanks{$^\dagger$ Corresponding and lead author.}

\author{M. Carmen Somoza}
\email{carmensomoza@uvigo.gal}

\address{ Departamento de Matemática Aplicada 1, Universidad de Vigo, Escola de Enxe\~ner\'ia Industrial, Rua Conde de Torrecedeira 86, CP 36208, Vigo, Spain \& CITMAGA, 15782 Santiago de Compostela, Spain.}

\begin{abstract}
A generalization of the notion of spanning tree congestion for weighted graphs is introduced. The $L^p$ congestion of a spanning tree is defined as the $L^p$ norm of the edge congestion of that tree. In this context, the classical congestion is the $L^\infty$-congestion. Explicit estimations of the minimal spanning tree $L^p$ congestion for some families of graphs are given. In addition, we introduce a polynomial-time algorithm for approximating the minimal $L^p$-congestion spanning tree in any weighted graph and another two similar algorithms for weighted planar graphs. The performance of these algorithms is tested in several graphs.
\end{abstract}

\maketitle
\pagestyle{myheadings}
\markleft{{\small A. Castej\'on Lafuente, E. Est\'evez, C. Meni\~no Cot\'on \& M. C. Somoza}}
\section{Introduction}

	The concept of congestion of a spanning tree on a connected graph was introduced by S. Simonson \cite{Simonson} in 1987 and revisited by  M. Ostrovskii in \cite{Ostrovskii2004}. The congestion of an edge in a spanning tree is the number of edges which are adjacent to the two connected components of the spanning tree obtained by removing the given edge. The congestion of a spanning tree is the maximum edge congestion.	The minimum spanning tree congestion (STC for short) of the graph is the minimum congestion among the family of spanning trees of the graph. 
	
	One recurrent question on the minimum STC of a graph is its computation. It is known that this problem is, in general, NP-hard \cite{Okamoto-Otachi-Uehara-Uno2011} and therefore a polynomial time algorithm for an exact computation of the minimal congestion, working on arbitrary graphs, is unexpected.  
	
	The minimum STC has been computed for several families of graphs. Some examples are: complete (multipartite) graphs \cite{Law-Ostrovskii} or planar and toric grids \cite{Hruska2008}, \cite{Ostrovskii2010}. On other families of graphs estimations are known as for cubical grids \cite{Castejon-Ostrovskii2009}, hypercube graphs \cite{Hiufai2009} or random graphs \cite{Ostrovskii2011, Chandran-Cheung-Isaac2018}. In addition, the minimal congestion of the duplication of a graph or the cartesian product of graphs have been partially studied in \cite{Ostrovskii2011}, \cite{Law-Ostrovskii2010},  \cite{Kozawa-Otachi}, \cite{Hiufai2009} or \cite{Kozawa-Otachi2018} .

	The notion of minimum STC was generalized for edge-weighted graphs in \cite{Otachi-Boadlaender-Leeuwen}, in that work the computational complexity of this problem is studied. In this context, the number of edges connecting two boundary components is replaced by the sum of the respective edge-weights, this generalization seems to be interesting for applications (a weight can be an observable in an applied problem: a distance between nodes, a cost, a travelling time, etc.). It was shown, in \cite[Lemma 7.1]{Otachi-Boadlaender-Leeuwen}, that the spanning tree congestion for integer valuated weighted graphs is equivalent to the classical congestion in an unweighted graph where every edge has multiplicity equal to its weight. For real valuated weighted graphs, one can make rational aproximations to get a similar result. However, in practice, this is not an efficient way to face the problem as, in general, very large multiplicities can arise from a generic real choice of weights. This implies that the STC problem for real valuated weighted graphs should be studied in its own right. 
	
	Moreover, in this work we also deal with another generalization of the STC problem. Instead of considering the maximum over the edge-congestions in the spanning tree, we can consider the sum of the edge congestions, this is what we call the $L^1$-congestion of the spanning tree. This is not a completely new invariant for unweighted graphs as it is equivalent to obtain a Low Stretch Spanning Tree as defined in \cite{Elkin-Emek}, as far as we know the given algorithm is also new in this context. A short proof of the Low Stretch of multipartite complete unweighted graphs is also given. 
	
	A minimal $L^1$-STC optimizes the mean edge-congestion and it seems to be a good option in band-with/cut-off problems on weighted graphs. The $L^p$-congestion of a spanning tree can be defined in a similar way for any $1<p<\infty$. It is shown that H\"older inequalities between $L^p$ norms induce natural inequalities between the $L^p$-STC's, this property can be exploited to improve the performance of the algorithm.
	
We also provide three algorithms for upper estimation of the minimum $L^p$-STC problem on weighted graphs. The first one works in any graph and the others only work in planar graphs. The first algorithm is a {\em congestion descent} algorithm, where each descending instance runs in polynomial time for a fixed weight function. The other algorithms run also in polynomial time, but we remark that one of them runs in polynomial with independence of the weight function. Despite the simplicity of the first algorithm, it produces reasonable results for the typical graphs where the congestion problem has been studied.
	
	Remark that an asymptotically good polynomial time algorithm for polylog (unweighted) graphs has been recently described in \cite{Kolman} and another interesting polynomial time algorithm was given in \cite{Chandran-Cheung-Isaac2018}. As far as we know, these algorithms have not been still implemented.
	
	Our algorithms\footnote{Meni\~no Cot\'on, C. (2025). Minimal Spanning tree estimation via descent methods (Version 1.0) [Computer software]. \url{https://github.com/Carlos-Menino/Minimal_Congestion_Spanning_Tree}} have a reasonable performance in the examples of graphs where the minimum spanning tree congestion is known or asymptotically estimated. Moreover, in several examples, the algorithm provides minimal spanning trees for both $L^\infty$ and $L^1$ congestion. These superoptimal trees, that do not need to exist in general, are better choices for applied problems. 
	
	As an example of the algorithm performance, the upper bounds provided by our algorithm for the STC problem in hypercubes of dimensions $8$, $9$ and $10$ improve the previously known bounds.

\section{Definitions and some estimates}
A {\em graph} $G$ is a pair $(V,E)$ where $V=\{v_1,\dots,v_n\}$ is called the set of {\em vertices} and $E\subset (V\times V)/\Z_2$ is the (multi)set of {\em edges}, formed by unordered pairs of vertices. A pair $\overline{uv}\in E$ is identified with an edge adjacent to $u$ and $v$. Formally $\overline{uv}$ can be interpreted as a set, henceforth adjacent vertices can be interpreted as elements of an edge. For a given graph $G$, let us denote by $V_G$ and $E_G$ its vertex and edge set respectively. An {\em edge-weighted} graph is a pair $(G,\omega)$ where $\omega:E_G\to\R^+$ is a positive function. Observe that an unweighted graph can be considered as an edge-weighted graph where every weight is equal to $1$.

It is said that $H$ is a {\em subgraph} of $G$ if $V_{H}\subset V_G$ and $E_{H}\subset E_G$. It is said that $H$ is a {\em spanning} subgraph of $G$ if it is a subgraph and $V_{H}= V_G$.

Let $v\in V_G$ be a vertex of a graph $G$. The set of {\em incident edges} is defined as 
$$I(v)=\{e\in E_G\,|\,v\in e\}$$
and the set of {\em adjacent vertices} is defined as 
$$A(v)=\{u\in V_G\,|\,\overline{uv}\in E_G\}\;.$$ 
The {\em degree} of a vertex $v$, denoted by $\deg(u)$, is the cardinal of $I(v)$. 

A graph is called {\em simple} if there are no edges of the form $\overline{vv}$ for any $v\in V_G$ and every pair of vertices are adjacent to at most a single edge.

A simple graph $G$ is {\em complete} if every pair of different vertices are adjacent to an edge. Complete graphs will be noted, as usual, by $K_n$ where $n$ is the number of vertices. In a similar way, a multipartite complete graph is a graph determined by a partition of the set of vertices, an edge exists if and only if  it connects vertices in different sets of the partition. Complete multipartite graphs will be noted by $K_{n_1,\dots,n_k}$ where $n_i\geq n_j$ are the sizes of the partition sets.

Let $u,v\in V_G$, a $(u,v)$-{\em path} is an ordered tuple of pairwise different edges $(e_1,\dots,e_\ell)$ such that $w_0=u\in e_1$, $w_{\ell}=v\in e_\ell$, $w_{i}=e_i\cap e_{i+1}\neq \emptyset$ for $i=1,\dots,\ell-1$ (assuming $k>1$) and $w_i\neq w_{j}$ for $i\neq j$. The {\em length} of a $(u,v)$-path  $\wp=(e_1,\dots,e_\ell)$ is $\ell$ and denoted by $|\wp|$. The {\em distance} between two vertices $u,v$, denoted by $d_G(u,v)$ is the minimum length among all the $(u,v)$-paths.

If $G$ is edge-weighted, then the weight of a $(u,v)$-path $\wp$, denoted by $\omega(\wp)$ is defined as the sum $\omega(e_1) + \cdots + \omega(e_\ell)$. The weighted distance between $u,v\in V_G$, denoted $d_{G,\omega}(u,v)$ is the minimum weight among all the $(u,v)$-paths. Analagously we can define the weight of any subgraph as the sum of the weights of its corresponding edges.
 
It is said that $G$ is {\em connected} if there exists at least an $(u,v)$-path for every pair of vertices $u\neq v$. A {\em cycle} is a subgraph that can be obtained from the union of a $(u,v)$-path and an edge $\overline{uv}$, it can be also called {\em closed path}. A simple and connected graph without cycles is a {\em tree}. A {\em spanning tree} of a graph $G$ is a spanning subgraph which is a tree. Of course, spanning trees do only exist on connected graphs. The set of spanning trees of a connected graph $G$ is denoted by $\mathcal{S}_G$.

In a simple graph $G$, each proper subset $X\subset V_G$ determines the partition $\{X,X^c\}$ where $X^c$ denotes the complement of $X$ in $V$. Let $T\in \mathcal \cS_G$ be a spanning tree. Let $\{X_e,X_e^c\}$ be the connected components of the graph $T_e = (V_G, E_T\setminus\{e\})$. Let $E(X)$ be the set of edges that are adjacent to both $X$ and $X^c$, the {\em width} of a vertex partition $\{X,X^c\}$ is $|E(X)|$. If $G$ is edge-weighted, the {\em weighted-width} of a vertex partition $\{X,X^c\}$, denoted by $e_\omega(X,X^c)$, is the sum of weights of edges in $E(X)$.

\begin{defn}[edge congestion]
Let $T$ be a spanning tree of an edge-weighted graph $(G,\omega)$. The {\em congestion} of an edge $e\in E_T$, denoted by $\cC(G,T,e,\omega)$ is defined as the weighted-width of the partition $\{X_e,X_e^c\}$.
\end{defn}

Observe that $(\cC(G,T,e,\omega))_{e\in E_T}$ can be considered as an element of $\R^{|V| -1}$ for every spanning tree $T$, this vector only depends on the ordering of the edges of $T$. 

\begin{defn}
Let $T$ be a spanning tree of an edge-weighted graph $(G,\omega)$. Let $f: \R^{|V| -1}\to\R^+$ be a symmetric function. The $f$-congestion of $T$ in $(G,\omega)$, denoted by $\cC_f(G,T,\omega)$, is defined as $f((\cC(G,T,e,\omega))_{e\in E_T})$. The spanning tree $f$-{\em congestion} of $(G,\omega)$, denoted by $\cC_f(G,\omega)$ is the minimum value of $\cC(G,T,\omega)$ for $T\in\cS_G$. If the minimum is attained at some tree $T$, then $T$ is called a {\em minimum $f$-congestion spanning tree} for $(G,\omega)$.

In this work we are only interested in the case of $L^p$-norms in $\R^{|V| -1}$ for $p\in [1,\infty]$, the $L^p$-congestion of the spanning tree $T$ is defined as:
$$
\cC_p(G,T,\omega)=\|(\cC(G,T,e,\omega))_{e\in E_T}\|_p\,.
$$

and $\cC_p(G,\omega)=\min_{T\in\cS_G}\cC(G,T,\omega)$ is called the {\em spanning tree $L^p$-congestion} of $G$ ($L^p$-STC for short). We shall remove $\omega$ from the notations in the case of unweighted graphs. 
\end{defn}

\begin{rem}
Observe that the classical STC of an unweighted connected graph is just the $L^\infty$-STC of that graph. Our first motivation to define $L^p$-congestions is the $p=1$ case, the $L^1$-congestion of a graph minimizes the mean of the edge-congestions of the spanning tree instead of the maximum one. It is clear that for several applications, minimizing the mean congestion (understanding congestion as a risk measurement) of the edges can be relevant, for instance, when the damage source is a random variable.

Moreover, we find that for several families of graphs there exist spanning trees minimizing both $L^p$-congestions for $p=1,\infty$. These superoptimal spanning trees should be considered as better solutions for the STC problem.

$L^p$-congestions for $1<p<\infty$ are obvious generalizations of the $L^1$-congestion, but they deserve their own interest: recall that $\|\cdot\|_p>\|\cdot\|_q$ for $p<q$ and $\|\cdot\|_p\rightarrow \|\cdot\|_\infty$ as $p\rightarrow \infty$. Thus, we can use $L^p$-congestions as another tool for obtaining estimates and bounds for the classical congestion.
\end{rem}

\begin{prop}\label{p:Holder}
Let $(G,\omega)$ be a connected edge-weighted (simple) graph. The following relation holds
$$
\dfrac{\cC_p(G,\omega)}{\sqrt[p]{|V|-1}}\leq \cC_\infty(G,\omega) < \cC_p(G,\omega)\;.
$$
The sequence $\{\cC_p(G,\omega)\}_{p\in\N}$ is monotonically decreasing and $\lim\limits_{p\to\infty}\cC_p(G) = \cC_\infty(G)$. Moreover, there exists a spanning tree $T$ that minimizes  the $L^p$-STC of $(G,\omega)$ for any sufficiently large $p$ and $p=\infty$.
\end{prop}
\begin{proof}
Let $T$ be a minimum spanning tree for the $L^p$ congestion. Set $q>p$, since $\|v\|_q <\|v\|_p$ for every vector with nonzero coordinates, it follows that $\cC_q(G,T,\omega)<\cC_p(G,T,\omega)=\cC_p(G,\omega)$. This shows that the sequence of $L^p$-congestions is decreasing and the right inequality $\cC_\infty(G) <\cC_p(G)$.

From the H\"older inequality we know that $\|v\|_p \leq \left(\sqrt[p]{|V|-1}\right)\cdot\|v\|_\infty$. Therefore, by a similar argument as before we get the left inequality.

Since $\sqrt[p]{|V|-1}$ goes to $1$ as $p$ goes to $\infty$, we get, from the given inequalities, the convergence of the $L^p$-congestions to the $L^\infty$-one.

For the last assertion just take into account that the number of spanning trees of $G$ is finite. Let $T_p$ be a minimal spanning tree for the $L^p$-STC of $G$ for each $p\in\N$. Since there are finitely many spanning trees, it follows that infinitely many of the $T_p$'s are equal. Let $T$ be such spanning tree. From the convergence of the $L^p$-STC to the $L^\infty$-one, it follows that $T$ is also a minimal spanning tree for the $L^\infty$-STC. It remains to show that $T$ is a minimal spanning tree for all sufficiently large $p$. We shall use the following claim, that is left as an exercise for the reader.

\begin{claim}\label{claim:norm}
Let $v,w\in\R^n$ arbitrary. There exists $N\in\N$ such that $\|v\|_p\leq\|w\|_p$ for all $p\geq N$ or $\|w\|_p\leq\|v\|_p$ for all $p\geq N$.
\end{claim}

Assume now that there is another tree $T'$ which is a minimal spanning tree for the $L^p$-STC for infinitely many values of $p$. Let $p_n,q_n$ sequences of natural numbers such that $\cC_{p_n}(G,\omega)=\cC_{p_n}(G,T,\omega)$ and $\cC_{q_n}(G,\omega)=\cC_{q_n}(G,T',\omega)$ for all $n\in\N$. It follows that $\cC_{p_n}(G,T,\omega)\leq \cC_{p_n}(G,T',\omega)$ and $\cC_{q_n}(G,T',\omega)\leq \cC_{q_n}(G,T,\omega)$. From the Claim~\ref{claim:norm} it follows that both congestions must be the same for sufficiently large $n$. It follows that $\cC_p(G,\omega)=\cC_p(G,T,\omega)$ for all sufficiently large $p$. 

\end{proof}

It is also interesting to note that the $L^p$-STC is stable under weight perturbations, more precisely:

\begin{prop}
Let $(G,\omega)$ be an edge weighted graph. Let $\omega_n:E_G\to\R^+$ be a sequence of weigths such that $\omega_n$ converges to $\omega$ with respect to the supremum norm. Then
$$
\lim\limits_{n\to\infty}\cC_{p}(G,\omega_n) = \cC_p(G,\omega)\,,\ \forall\ p\,,\ 1\leq p\leq \infty\;.
$$
\end{prop}

As a final remark, note that $\cC_p(G,\omega)\leq \cC_p(G,\lambda)$ if $\lambda(e) \geq \omega(e)$ for all $e\in E$. In particular $m\cdot\cC_p(G)\leq \cC_p(G,\omega)\leq M\cdot\cC_p(G)$ where $m$ and $M$ are, respectively, the minimum and maximum values of $\omega$ and $\cC_p(G)$ is the $L^p$-STC of the respective unweighted graph $G$.

\subsection{Known results for the STC problem}\label{s:known results}

This subsection is devoted to give a brief summary of general estimates and explicit computation of the classical congestion for some families of graphs.

Let us recall first some generic upper and lower bounds for the classical spanning tree congestion. Following \cite[Proposition 6.1]{Kozawa-Otachi-Yamazaki2009} and \cite[Theorem 1]{Ostrovskii2004}, it is known that, for every spanning tree $T$ of a graph $G=(V,E)$ we have:
$$
\dfrac{2|E|}{|V|-1}-1\leq\cC_\infty (G,T)\leq |E|-|V|+2
$$
$$
\cC_\infty (G,T)\leq\left\lceil\dfrac{|V|}{2}\right\rceil \cdot \cC_\infty(G)\;.
$$
Both inequalities can be used to give lower bounds for the classical congestion. One interesting point, relative to the second inequality, is that a spanning tree of maximum congestion provides a lower bound for the minimum spanning tree congestion. In general, these are poor estimates for the STC problem.

For $u,v\in G$, let $m(u,v)$ be the maximum number of edge disjoint $(u,v)$-paths in $G$ and let 
$$
\phi(G)=\min_{|X|\leq|V|/2}\dfrac{e(X,X^c)}{\sum_{v\in X}\deg(v)} 
$$
be the Cheeger constant of $G$. Let $\Delta$ and $\delta$ be the maximum and minum degree of vertices in $G$.

\begin{prop}\cite{Ostrovskii2004}\label{p:lower_bounds}
For every connected graph $G=(V,E)$ we have:
\begin{itemize}
\item $\max_{u,v\in V}m(u,v)\leq \cC_\infty(G)$

\item $\delta\cdot\phi(G)\cdot\max\left(\dfrac{|V|-1}{\Delta},\left\lfloor\dfrac{\diam(G)}{2}\right\rfloor\right)\leq \cC_\infty(G)$
\end{itemize}
\end{prop}

The STC problem has been solved or estimated for several families of graphs, these will be called {\em test graphs}:

\begin{itemize}
\item The STC of complete graphs and  complete bipartite graphs were computed in \cite{Hruska2008}, complete multipartite graphs were handled in \cite{Law-Ostrovskii}.

\item Let $R_{m,n}$ denote the rectangular grid with $m\times n$ nodes and let $T_n$ be the triangular grid with $n$ nodes on each side. The STC of these graphs were computed in \cite{Hruska2008} and \cite{Ostrovskii2010}.

\item Hexagonal grids were studied in \cite{Estevez}: $\rm{Hex}_n$ will denote a triangular pattern of hexagons, where $n$ is the number of hexagons in each side; and $\rm{Hex}_{n,m}$ will denote a rectangular pattern of hexagons where $n,m$ is the number of side hexagons.

\item Let $H_d$ denote the hypercube graph, its congestion was studied in \cite{Hiufai2009}. 

\item The STC of some specific toroidal grids, more precisely, those which are the Cayley graph of $\Z_m\times \Z_n$ and $\Z_k^3$, were studied in \cite{Castejon-Ostrovskii2009,Kozawa-Otachi-Yamazaki2009}.

\item The STC was estimated for the random graphs $\mathcal{G}(n,p)$ in \cite{Chandran-Cheung-Isaac2018, Ostrovskii2011}.
\end{itemize}

The STC estimations of these test graphs are summarized in Table~\ref{tab:STC}.

\begin{table}

\caption{Spanning tree congestions of test graphs.}

\centering
\begin{tabular}{|c|c|}
\hline
\label{tab:STC}
Graph	&	$L^\infty$ congestion 	\\ \hline
$K_n$	&	$n-1$		\\ \hline
$K_{n,m}$($2\leq m\leq n$)	&	$m+n-2$		\\ \hline
$R_{m,n}$ $n>m>3$, $m$ even	&	$m+1$		\\ \hline
$R_{m,n}$ $n\geq m>3$ or $m$ odd	&	$m$			\\ \hline
$T_n$ ($n\geq 3$)	&	$2\cdot \lfloor 2n/3)\rfloor$		\\ \hline
$\rm{Hex}_n$	&	$\leq 3+ \lfloor (n-1)/3\rfloor + \lfloor (n-2)/3\rfloor$	\\ \hline
$\rm{Hex}_{m,n}$, $m\geq n$	&	$\leq 1+2\cdot\lceil n/2\rceil$	\\ \hline
$H_d$ ($d<7$)	&	$2^{d-1}$	\\ \hline
$H_d$ ($d\geq 7$)	& $\in \left[\dfrac{\log_2(d)}{d}2^d,\dfrac{\lfloor\log_2(d)\rfloor}{2^{\lfloor\log_2(d)\rfloor}}2^d+ d^2 - 2\lfloor \log_2(d)\rfloor d + O(\log_2(d))\right]$		\\ \hline
$\Z_m\times \Z_n$ &  $2\min\{m,n\}$  \\ \hline
$\Z_k^3$ & $\in\left[(2/\sqrt{6})k^2-o(1), (7/4)k^2+ O(k) \right]$   \\ \hline
$C_{k,k,k}$ & $\in\left[(2/\sqrt{6})k^2-o(1), (\sqrt{15}-3)k^2+ O(k)\right]$  \\ \hline
$K(n_1,...,n_k)$ if $ n_k>1$& $2N-n_1-n_2-2$ \\ \hline
$K(n_1,...,n_k)$ if $ n_k=1$& $N-n_{k-1}$  \\ \hline
$\mathcal{G}(n,p), p> c\log(n)/n$&$ \Theta(n)$ w.h.p. for $c\geq 64$ \\ \hline
\end{tabular}

\end{table}

\subsection{Lower bounds for the $L^p$-STC}
The aim of this section is to obtain general and computable lower bounds for the $L^p$ congestion. First of all, we will adapt Proposition~\ref{p:lower_bounds} to the case of edge-weighted graphs.

Let $(G,\omega)$ be an edge-weighted graph. Let $\wp$ be a path in $(G,\omega)$, define $m_\omega(\wp)$ as the minimum edge-weight among the edges in the path $\wp$. Let $\mathcal{P}(u,v)$ be the family of maximal sets formed by pairwise edge-disjoint $(u,v)$-paths (i.e., each set has the maximum number of such paths). Let us define
$$
m_\omega(u,v) = \max_{M\in\mathcal{P}(u,v)}\left\{\sum_{\wp\in M}m_\omega(\wp)\right\}\;.
$$
Observe that $m_\omega(u,v)=m(u,v)$ if $\omega$ is constant equal to $1$.

\begin{prop}\label{p:weighted_bound_infty}
Let $(G,\omega)$ be an edge-weighted graph. Then
$$
\max_{u,v\in V} m_\omega(u,v)\leq\mathcal{C}_\infty(G,\omega)\;.
$$
\end{prop}
\begin{proof}
Let $u,v\in V$ and let $T$ be an arbitrary spanning tree. Let $e$ be an edge in $T$ so that $u\in X_e$ and $v\in X_e^c$. Let $M\in\mathcal{P}(u,v)$. For every $(u,v)$-path $\wp$ in $C$ we obtain at least one edge $e_\wp$ joining $X_e$ with $X_e^c$. The weight of this particular edge in $\wp$ must be geater or equal than $m_\omega(\wp)$ by definition. Since the paths in $M$ are pairwise edge-disjoint, the edges $e_\wp$ are also different and therefore $\cC(G,T,e,\omega)\geq \sum_{\wp\in M}m_\omega(\wp)$.
\end{proof}

We want to adapt the previous idea for the $L^p$-STC. Let $e = \overline{uv}$ be an edge of an edge-weighted connected graph $(G,\omega)$ and define $m_\omega(e) = m_\omega(u,v)$.

\begin{lemma}\label{l:path_inequality}
Let $(G,\omega)$ be an edge-weighted connected graph and let $T$ be any spanning tree. Let $e$ be an edge of $T$ then 
$$
m_\omega(e)\leq \cC(G,T,e, \omega)\;.
$$ 
\end{lemma}
\begin{proof}
Assume $e=\overline{uv}$, it follows that $u\in X_e$ and $v\in X_e^c$. Therefore, by the same argument given for Proposition~\ref{p:weighted_bound_infty}, we get $m_\omega(e)\leq \cC(G,T,e,\omega)$ as desired.
\end{proof}

\begin{prop}\label{p:m bound}
Let $(G,\omega)$ be an edge-weighted graph. Let $T_\ast$ be a Kruskal minimal spanning tree for the edge-weighted graph $(G, m_\omega)$ then
$$
\|(m_\omega(e))_{e\in E_{T_\ast}}\|_p\leq \cC_p(G,\omega)\;.
$$
\end{prop}
\begin{proof}
The Kruskal algorithm produces a minimal spanning tree by adding the minimal edges under the condition of not generating a cycle. Let $(e_1,e_2,\dots,e_{|V|-1})$ be an increasing ordering of the edges in $T_\ast$ respect to $m_\omega$. If $T$ is any other spanning tree, we can also order its edges in a similar way: $(e'_1,e'_2,\dots,e'_{|V|-1})$. The Kruskal algorithm guarantees that $m_\omega(e_i)\leq m_\omega(e'_i)$ for all $i$ and therefore $\|(m_\omega(e))_{e\in T_\ast}\|_p\leq \|(m_\omega(e))_{e\in T}\|_p$. But $m_\omega(e)\leq \cC(G,T,e,\omega)$ by Lemma~\ref{l:path_inequality}, this gives the inequality.
\end{proof}

If $G$ is unweighted then $m(e)$ can be computed in polynomial time by means of the Ford-Fulkerson algorithm \cite{Ford-Fulkerson}. For a weighted graph it is unclear if the computation of $m_\omega(u,v)$ can be done in polynomial time (the families of paths in $\cP(u,v)$ can be very large in general). In any case, since $m_\omega$ is a maximum it is rather easy to get reasonable lower bounds for $m_\omega(u,v)$.

For weighted graphs there is also a notion of weighted Cheeger constant.
Define the {\em weighted degree} of $v$ as $\deg_\omega(v)=\sum_{e\in I(v)} \omega(e)$ and $\delta_\omega = \min_{v\in V}\deg_\omega(v)$. For $X\subset V$, let $\vol_\omega(X)=\sum_{v\in X}\deg_\omega(v)$. The {\em weighted Cheeger constant} is defined as the maximal number $\phi_\omega(G)$ satisfiying
$$
e_\omega(X,X^c)\geq \phi_\omega(G)\cdot\min(\vol_\omega(X),\vol_\omega(X^c)) 
$$
for evey $X\subset V_G$.

Following the same argument as in \cite[Theorem 1(b)(c)]{Ostrovskii2010} we obtain 

$$
\cC_\infty(G,\omega)\geq \dfrac{\delta_\omega}{\Delta}\phi_\omega(G)\cdot\max(|V|-1,\diam(G)/2)
$$

For the $L^p$ congestion we need a local notion of Cheeger constant, namely $\phi_\omega(G,e)$ and defined as the maximal number satisfying the condition
$$
e_\omega(X,X^c)\geq \phi_\omega(G,e)\cdot\min(\vol_\omega(X),\vol_\omega(X^c)) 
$$
for every $X\subset V$ such that $X$ and $X^c$ are both adjacent to $e$. Observe that $\phi_\omega(G)=\min_{e\in E_G} \phi_\omega(G,e)$. From this localized Cheeger constant we get a lower bound for the edge congestion of an edge which is independent of any chosen minimal tree $T$ that contains that edge: 
$$\cC(G,T,e,\omega)\geq \delta_\omega\phi_\omega(G,e)\cdot\max\{(|V|-1)/\Delta,\lfloor\diam(G)/2\rfloor\}\;.$$

\begin{prop}
Let $(G,\omega)$ be an edge-weighted graph. Let $T_\ast$ be a Kruskal minimal spanning tree for the edge-weighted graph $(G, \phi_\omega(G,e))$ then
$$
\delta_\omega\|(\phi_\omega(G,e))_{e\in E_{T_\ast}}\|_p\cdot\max\left\{\dfrac{|V|-1}{\Delta},\left\lfloor\dfrac{\diam(G)}{2}\right\rfloor\right\}\leq \cC_p(G,\omega)\;.
$$
\end{prop}

\begin{proof}
It is analogous to Proposition~\ref{p:m bound}, just consider a Kruskal tree in $(G,\phi_\omega(G,\cdot))$ and the previous lower bound of the edge congestion by means of the localized weighted inequality. The term $\delta_\omega\max\{(|V|-1)/2,\diam(G)/2\}$ is just a constant term and does not affect to the minimal spanning tree computation, therefore it can be handled by the homogeneity of the norm.
\end{proof}

We include two final results that are also of independent interest.

\begin{prop}
For every edge weighted connected simple graph $(G,\omega)$ and every spanning tree $T$ we have:
$$
\sum_{e\in E_T}\omega(e) + 2\cdot\sum_{e\notin E_T}\omega(e)=2\cdot\omega(G)-\omega(T)\leq \cC_1(G,T,\omega)\;.
$$
\end{prop}
\begin{proof}
Let $e=\overline{uv}\in E_G$ arbitrary and let $T$ be any spanning tree. There exists a unique $(u,v)$-path $\wp$ in $T$. It follows that every edge in $\wp$ separates $u$ from $v$ and, hence, $\omega(e)$ is counted once for the edge congestion computation of each edge in $\wp$. Observe that $\wp$ has length $1$ if $e\in E_T$ and length at least $2$ if $e\notin E_T$ and the result follows.
\end{proof}

\begin{cor}\label{c:edge counting bound}
Let $(G,\omega)$ be an edge weighted graph $(G,\omega)$. Let $T_\ast$ be a maximal spanning tree of $(G,\omega)$, then
$$
\sum_{e\in E_{T_\ast}}\omega(e) + 2\cdot\sum_{e\notin E_{T_\ast}}\omega(e)=2\cdot\omega(G) - \omega(T_\ast)\leq \cC_1(G,\omega)\;.
$$
\end{cor}

In the case of an unweighted graph, i.e. $\omega = 1$, Corollary~\ref{c:edge counting bound} gives the bound $(|V|-1) + 2(|E|-|V|+1) = 2|E| - |V| + 1\leq \cC_1(G)$. This provides, via the H\"older inequality given in  Proposition~\ref{p:Holder}, the known lower bound given at the beginning of section~\ref{s:known results} for the STC problem.

At this point, we have to recall that the minimum $L^1$-STC problem on unweighted graphs is equivalent to minimizing the total stretch of a spanning tree (this was observed in \cite[Pag. 3]{Chandran-Cheung-Isaac2018}). Just for completeness, we include the equivalence between both problems and observe that this is no longer true for weighted graphs.

\begin{defn}[Total stretch of a spanning tree (see \cite{Elkin-Emek})]
Let $G=((V,E),\omega)$ be an edge-weighted graph. Let $T$ be a spanning tree of $G$. The stretch of an edge $e\notin E_T$ is defined as:
$$
\sigma(G,T,e,\omega)=\dfrac{\omega(\wp_e)}{\omega(e)}\,,
$$
where $\wp_e$ is the unique path in $T$ joining the adjacent vertices of $e$. The {\em (total) stretch} of $T$ in $(G,\omega)$ is defined as:
$$
\sigma(G,T,\omega) = \sum_{e\in E_G\setminus E_T} \sigma(e,T,\omega)\;;
$$
The low stretch of $G$, denoted as $\sigma(G,\omega)$ is the total stretch of minimal total stretch. Both $G$ and $\omega$ can be omitted from the notations when $G$ is clear in the context and it is unweighted.
\end{defn}

 There exists an algorithm for the LSST problem on weighted graphs, see \cite{Elkin-Emek}, that gets a spanning tree with stretch $O(|E|\log(|V|)^3)$ that runs in $O(|E|\log(|V|))$ time.

\begin{prop}\label{p:LSSTvsSTC}
Let $G$ be an unweighted graph and let $T$ be a spanning tree. Then $T$ minimizes the $L^1$-STC if and only if $T$ is a Low Stretch Spanning Tree. 
\end{prop}  
\begin{proof}
Given an edge $e=\overline{uv}\in E_G\setminus E_T$, let $\wp_e$ be the unique $(u,v)$-path in $T$. As above, observe that $e\in E(X_a)$ for some $a\in E_T$ if and only if $a\in \wp_e$. Since the $L^1$-congestion is the sum of all edge congestions in $T$, it follows that $e$ is counted $|\wp_e|$ times for the computation of $\cC_1(G,T)$ and therefore $$\cC_1(G,T) = |E_T| + \sum_{e\in E_G\setminus E_T}|\wp_e|\;.$$  Since $G$ is unweighted we identify it with a weighted graph with constants weights equal to $1$, thus $|\wp_e| = \sigma(T,e)$. Hence $\cC_1(G,T) = |V|-1 + \str(G,T)$, and the $L^1$-congestion and the total stretch of a spanning tree only differs by a constant term.
\end{proof}

In the case of edge-weighted graphs, a similar argument yields $\cC_1(G,T,\omega) = \sum_{e\in E_T}\omega(e) +\sum_{e\in E_G\setminus E_T}|\wp_e|\cdot\omega(e)$. Since $\omega$ does not need to be constant, there is no direct relation between $|\wp_e|\cdot \omega(e)$ and $\omega(\wp_e)/\omega(e)$. Thus, minimal spanning trees for $L^1$-congestion and weighted stretch can be different. For instance, let us consider the complete graph in three vertices $a,b,c$ with edge weights $\omega(\overline{ab})=1$, $\omega(\overline{bc})=1$, $\omega(\overline{ac})=2$, an easy computation yields that $\{\overline{ab},\overline{bc}\}$ has minimum LSST and $\{\overline{ab},\overline{ac}\}$ has minimum $L^1$-STC but not minimum stretch.

As a final comment in this section let us pose a question relative to the relation of congestions of graphs and subgraphs. It is known that the congestion of a subgraph can be larger than the congestion of the whole graph (see \cite[Lemma 5]{Kolman}) however this question is more interesting for spanning subgraphs, i.e., subgraphs with the same number of vertices as the total graph. This can be translated to the following question: is it possible to add edges to a graph an get a supergraph with lower STC?

The answer to this question is affirmative, but we still do not know (nontrivial) necessary conditions for this to happen. The following result is is implicit in several works in the literature, see e.g. \cite[Proposition 5.3]{Otachi-Boadlaender-Leeuwen}. 

\begin{prop}\label{p:better supergraph}
Let $G$ be an unweighted connected graph. There exists a supergraph $\tilde{G}$ with the same vertices as $G$ such that $\cC_\infty(\tilde{G})\leq \Delta + 1$.
\end{prop}
\begin{proof}
Let $v\in V_G$ be an arbitrary vertex. For every $u\in V_G$ add every edge $\overline{uv}$ that was not already in $E_G$, this supergraph will be called $\tilde{G}$. Let $\tilde{T}$ be the spanning tree formed by every edge in $\tilde{G}$ incident at $v$. Every vertex different to $v$ is terminal and hence the congestion of its adjacent edge is equal to its degree. The result follows.
\end{proof}

Thus, if $\cC_\infty(G) > \Delta + 1$, then there exists a supergraph with the same number of vertices as $G$ and lower congestion. An analogous result for weighted graphs is also possible but it depends in the weights of the new edges in the supergraph. The proof is analogous to the proof of Proposition~\ref{p:better supergraph}.

\begin{prop}Let $(G,\omega)$ be an edge-weighted connected graph. Let $(\tilde{G},\tilde{\omega})$ be a supergraph of $G$ with the same number of vertices as $G$. Suppose that there is a radial spanning tree whose center is some $v\in V_G$, then $\cC_\infty(\tilde{G},\tilde{\omega})\leq \max_{u\neq v}\{\deg_\omega(u) + \tilde{\omega}(\overline{uv})\}$.
\end{prop}
We think that a similar result as the given in Proposition~\ref{p:better supergraph} should hold for the $L^p$-STC problem. The study of the minimal number of edges needed to reduce the $L^p$-STC is left for a future work.

\subsection{$L^p$-STC of complete (multipartite) graphs}\label{s:complete_graphs}

In this subsection we deal with a specific family of simple unweighted graphs. These are the complete and complete bipartite graphs. We also give an upper bound for any complete multipartite graph that we conjecture to be tight.

The complete graph $K_n$ is easy to handle. If an edge $e$ of a spanning tree $T$ produces a partition of $G$ in two sets of size $m$ and $k$ (with $m+k = n$), then its congestion is easy to compute since every edge between these sets must exist in the complete graph:
$$
\cC(K_n,T,e)=|E(K_{m,n-m})|= m\cdot (n-m) = m\cdot n -m^2
$$

It is clear that $m\cdot n - m^2 \geq n-1$ for all $1\leq m\leq n-1$ (assuming $n>1$). Henceforth the edge congestion of any edge of any spanning tree is lower bounded by $(n-1)$. A radial spanning tree, i.e., a tree with a root vertex connected by an edge to any other vertex, minimizes the edge congestion of every edge and therefore is a minimal spanning tree for the $L^p$-congestion for every $p\in\N$. As a consequence we have:

$$
\cC_p(K_n)=\left( (n-1)^{p+1} \right)^{1/p}
$$

The $L^p$-STC of complete bipartite graphs $K_{m,n}$, $m\leq n$, is more difficult to obtain. In this section we compute its $L^1$-congestion, recall that this is equivalent to minimize the LSST problem.

Let $T$ be a spanning tree in $K_{m,n}$ and suppose that $e\notin E_T$ and let $\wp_e$ be the unique path in $T$ joining the adjacent vertices of $e$. It is clear that $|\wp_e|\geq 3$.

Let $T_\star$ be a spanning tree in $K_{m,n}$ such that every edge except one is terminal (only two vertices have degree $>1$). This spanning tree is unique up to isomorphism.

It is also clear that $|\wp_e| = 3$ for every $e\notin E_{T_\star}$, therefore $T_\star$ is a minimum LSST. By Proposition~\ref{p:LSSTvsSTC}, $T_\star$ is also a minimal $L^1$-STC. Thus
$$
\cC_1(K_{m,n})= 3nm - 2(n+m-1)\;.
$$

\begin{rem}
Observe that $T_\star$ minimizes the $L^1$-STC but not the classical STC problem. However, the spanning tree presented in \cite{Hruska2008} minimizes both the $L^1$ and $L^\infty$-STC's simultaneusly.
\end{rem}

For the multipartite case, let $X_1,\dots, X_k$ be the vertex partition associated to the complete multipartite graph $G=K_{n_1,\dots,n_k}$, where $|X_i|=n_i$ and $n_1\geq \dots\geq n_k$, set $N = n_1 + \dots + n_k$. 

Recall that a {\em terminal} vertex is a vertex of degree $1$ and a {\em terminal} edge is an edge adjacent to a terminal vertex.

The congestion of a terminal edge $e$ of a spanning tree whose adjacent terminal vertex belongs to $X_i$ is $N - n_i$. It is easy to see that every nonterminal edge has congestion larger than $N$, this fact does not depend on the spanning tree $T$.

For the specific case where $n_k=1$, there exists a (unique up to isomorphism) radial spanning tree such that every edge is terminal (precisely the one rooted at the unique point in $X_k$). By the previous affirmation this is a minimal $L^p$-congestion  spanning tree for every $p\in\N$.

More precisely, for $n_k=1$, we have $$\cC_p(K_{n_1,\dots,n_{k-1},1}) = \left( \sum_{i=1}^{k-1}n_i(N-n_i)^p\right)^{1/p}\;.
$$

When $n_k>1$ the situation is more complicated. In order to obtain a reasonable upper bound we consider two families of spanning trees, one of them is formed by trees similar to the one that minimizes the $L^\infty$-STC(as shown in \cite{Law-Ostrovskii}) and the other is formed by trees with just one nonterminal edge.

Let $1\leq j < i\leq k$. Define $T_M(i,j)$ as the spanning tree that connects a point of $X_i$ with every other point in $V\setminus X_i$, the remaining points of $X_i$ are connected, in one to one correspondence, with $n_i-1$ points in $X_{j}$.

Define $T_S(i,j)$ as the spanning tree that connects a point of $X_i$ with every other vertex in $V\setminus(X_i\cup X_j$) and with $n_{j}-1$ points of $X_{j}$, the remaining point on $X_{j}$ is connected to the remaining vertices in $X_i$.

Observe that each $T_S(i,j)$ has a unique nonterminal edge (with relatively large congestion) and each $T_M(i,j)$ has exactly $n_{j}-n_i +1$ nonterminal edges (with relatively small congestion). These should give good upper bounds for the $L^p$-congestion.

A direct computation yields:

$$\cC_p(G,T_S(i,j))^p = \sum_{\ell=1}^{k}n_\ell(N-n_\ell)^p - (N-n_{j})^p - (N-n_i)^p
+ (n_i(N-n_i-1) + 2 - n_{j})^p$$

and

$$\cC_p(G,T_M(i,j))^p = \sum_{\ell=1}^{k}n_\ell(N-n_\ell)^p - (n_i-1)(N-n_{j})^p - (N-n_i)^p + (n_i-1)(2N-n_i - n_{j}-2)^p\;.$$

For $p=1$ these expressions give the same value, i.e., for any choice of $(i,j)$ the $L^1$-congestions of $T_S(i,j)$ and $T_M(i,j)$ are the same. For $p>1$, this is not longer true, being $\cC_p(G,T_M(i,j))\leq \cC_p(G,T_S(i,j))$, that was expected as for $p$ big enough $T_S(i,j)$ do not give minimum STC. 

Among the $T_M(i,j)$'s, $1\leq j < i\leq k$, the minimum value of $\cC_p(G,T_M(i,j))$, for each $p\geq 1$, is achieved for $j=1$. This is clear since, for any fixed $n_i$, the expression of $\cC_p(G,T_M(i,j))$ is decreasing in $n_j$ for $n_k\leq n_j\leq n_1$. There is strong evidence (numerical and analytical) that the minimum is, in fact, obtained for $i=k$ or $i=2$ but we still do not have a complete proof for this conjecture. In any case, $T_M(2,1)$ minimizes the classical congestion and henceforth it minimizes the $L^p$-STC problem for large $p$. The number $\min\{\cC_p(G,T_M(k,1)),\cC_p(G,T(2,1))\}$ is used as a good upper bound for the $L^p$-STC problem on any complete multipartite graph. 

The $L^1$-STC of unweighted complete multipartite graphs can be computed from its equivalence with the LSST problem. We did not find an explicit reference to the computation of the low stretch of complete multipartite graphs, we present here a proof for completeness.

\begin{prop}\label{p:multipartite_avcong}
For $k\geq 3$ and $n_k>1$, we have:
\begin{align*}
\cC_1(K_{n_1,\dots,n_k})=\cC_1(K_{n_1,\dots,n_k},T_S(k,1))&=2\left(\sum_{i=1}^{k-2}\sum_{j=i+1}^{k-1} n_in_j\right) + 3(n_k-1)(N-n_k-1) + N-1\;.
\end{align*}
\end{prop}
\begin{proof} 
We shall find a minimum LSST. In the multipartite case for a given spanning tree $T$ and $e\notin E_T$ we have that $|\wp_e|\geq 2$. However, when $n_k>1$ is not possible to get a spanning tree where every edge has stretch $2$ (i.e., a radial one). Define $s_3(T)$ as the maximal number of edges with average stretch $\geq 3$. For any vertex $x$, define $\iota(x)\in\{1,\dots,k\}$ as the unique index so that $x\in X_{i(x)}$. Let $E_2(T)$ be the set of edges with stretch $2$.

We shall define, momentarily and for the sake of readability, $\sigma(e)=3$ for each edge in $E_T$. This will not affect to the computation of the LSST since every spanning tree has $N-1$ edges, this definition changes the stretch of a spanning tree by a constant term.

Suppose that $e =\overline{xy}$ is an edge whose stretch is $2$. Then, it must exist another vertex $z = z_e$, called the {\em node} of $e$, such that $\overline{zx},\overline{zy}\in E_T$. Then, for every $w\in X_{\iota(z)}$ we have one of the following possibilities:
\begin{itemize}
\item[(a)] Both edges $\overline{wx}$ and $\overline{wy}$ have stretch $\geq 3$. Observe that, if $\overline{wx}\in E_T$ and $w\neq z$ then $\overline{wy}\notin E_T$ and $\sigma(\overline{wy})=3$.

\item[(b)] If $\sigma(\overline{wx},T)=2$, then $\overline{wy}\notin E_T$ and $\sigma(\overline{wy},T)=4$ since $z$ and $w$ cannot be adjacent to the same edge. Thus the average stretch of both edges is $\geq 3$.
\end{itemize}

In the above situation, define $E^e(x)=\{\overline{xw}\mid\ w\in X_{\iota(z)}\}$ and $E^e(y)$ similarly. From (a) and (b), it follows that the edges in $E^e(x)\cup E^e(y)$ have average stretch $\geq 3$.

Let us make an important remark here, suppose that $\overline{xy}$ and $\overline{yq}$ are edges with stretch $2$, their nodes $z,z'$ belong to the same partition set and for some $w\in X_{\iota(z)}\setminus\{z,z'\}$ we have $\sigma(\overline{wx},T)=2$ (this is the extreme case (b)), then $\overline{wq}\notin E_T$ and $\sigma(\overline{wq},T)\geq 4$. This will be called the {\em Property}~$\star$. Property $\star$ has an interesting consequence:
\begin{claim}\label{cl:star property}
 let $e=\overline{xy}$ and $e'=\overline{pq}$ be different edges with strech $2$ and let $v\in X_{\iota(z_e)}, w\in X_{\iota(z_{e'})}$ be vertices so that $\overline{xv}$ and $\overline{qw}$ have stretch $2$, then $\overline{yv}$ and $\overline{pw}$ are different edges with stretch $\geq 4$.
\end{claim}
\begin{proof}[Proof of Claim~\ref{cl:star property}]
That $\sigma(\overline{yv},T)=4=\sigma(\overline{pw},T)$ is a consequence of (b). If $\overline{yv}=\overline{pw}$, then $y=p$ and $v=w$, in particular the nodes of $\overline{xy}$ and $\overline{pq}$ belong to the same partition set. But in this case we have a configuration where $\overline{xy},\overline{yq} $ have stretch $2$, their nodes belong to the same partition set and there is $w\in X_{\iota(z_e)}$ so that $\sigma(\overline{xw},T)=2 =\sigma(\overline{qw})$, this is forbidden by Property $\star$.
\end{proof}
Let us compute a lower bound for the number of edges in $s_3(T)$. Let us enumerate the vertices in $K_{n_1,\dots,n_k}$ as $x_{ij}$ where $1\leq i\leq k$ and $1\leq j\leq n_i$, we assume $x_{ij}\in X_i$.

Set $1\leq i\leq k-1$ and $1\leq j\leq n_i$ arbitrary. Suppose first that there exists an edge $e_{ij}$ with stretch $2$ adjacent to $x_{ij}$ whose node belong to $X_{i+1}\cup\cdots\cup X_k$. In this case $E^{e_{ij}}(x_{ij})$ is a set of edges with average stretch $3$ with at least $n_k$ elements, define $E(x_{ij}) = E^{e_{ij}}(x_{ij})$. If $x_{ij}$ is not adjacent to any edge with stretch $2$ and node in a partition set $X_{\ell}$ for some $\ell >i$, then every edge from $x_{ij}$ to a vertex in $X_{i+1}\cup\cdots\cup X_k$ has stretch $\geq 3$, these are $N-(n_1+\dots+n_i)\geq n_k$ edges with stretch $\geq 3$, these edges will define the set $E(x_{ij})$ in this case. In any case, for $x\neq y\in X_1\cup\dots\cup X_{k-1}$ we have $E(x)\cap E(y)=\emptyset$ by definition.

Define $$N_T=\bigcup_{i=1}^{k-1}\bigcup_{j=1}^{n_i} E(x_{ij})\;,$$
and recall that it depend in the choice of $e_{ij}$ for each $x_{ij}$.

It is clear that $|N_T| \geq (N-n_k)\cdot n_k$. It is not clear (and can be false in general) that $N_T$ has average stretch $\geq 3$ (by means of the extreme case (b)). However, the Claim~\ref{cl:star property} implies that any edge with stretch $2$ in $E^{\overline{xy}}(x)$ is univocally paired with one edge with stretch $\geq 4$ in $E^{\overline{xy}}(y)$. If our inductive construction added that edge in $E(y)$ then the average stretch of these two edges will be $3$ and there is nothing to fix, in other case we can add that edge of stretch $\geq 4$ not considered before to $N_T$. After adding these edges the average stretch of the final set will be $\geq 3$. 

Thus, by removing now the edges in $T$ that were artificially added in the previous construction, it follows that $s_3(T)\geq (N-n_k)\cdot n_k-(N-1) = (N-n_k-1)(n_k-1)$.

The spanning tree $T_S(k,1)$ satisfies that $s_3(T_S(k,1))=(N-n_k-1)(n_k-1)$ and therefore is minimal for both the LSST and $L^1$-STC problems. The proof is complete.
\end{proof}

\section{First Algorithm: Congestion Descent}

The idea of the first algorithm is to make small perturbations to an initial spanning tree in an edge-weighted connected graph $(G,\omega)$. Each perturbation is made in such a way that the $L^p$ congestion never increases. The algorithm will stop when no more congestion-reducing perturbations are possible.

Recall that every connected graph $G$ defines another graph $\cS_G$ whose vertices are the spanning trees of $G$, and two spanning trees are adjacent if and only if their edge sets are equal up to removing exactly one edge from each tree. Observe that $\cS_G$ is also connected.

The spanning tree congestion problem can be seen as an optimization of a function defined on $\cS_G$. Our algorithm is an analogue to a gradient descent algorithm adapted to our specific problem: from any spanning tree there is a path in $\cS_G$ to a minimal $L^p$-congestion spanning tree, our algorithm will find a path in $\cS_G$ to a local minimal congestion spanning tree in polynomial time. For the sake of simplicity, let $m = |E_G|$ and $n = |V_G|$.

Observe that the computation of the $L^p$-congestion, of a single spanning tree, can be done in $O(m\cdot n)$ time since any edge $\textbf{e}=\overline{uv}$, not in the tree, contributes only to the edge congestion of the edges in the unique path in $T$ joining $u$ and $v$, for each such edge $\textbf{e}$ we only need to compute that path in, at most, $O(n)$ time\footnote{For a tree $T$, finding a path between two points can be done in $O(\diam(T))$, typically $\diam(T)\ll n$.}.

Let $T$ be an arbitrary spanning tree and let $\bf{e}$ be an edge in $E_G\setminus E_{T}$. Let $\wp_{\bf{e}}=\{e_1,\dots, e_\ell\}$ be the unique path in $T$ joining the extreme points of $\bf{e}$.

\begin{prop}\label{p:equal_cong}
Let $T_{e\to\bf{e}}$ be the spanning tree obtained from $T$ by removing the edge $e\in\wp_{\bf{e}}$ and adding the edge $\bf{e}$. Then $T_{e\to\bf{e}}$ is a spanning tree and $\mathcal{C}(G,T,a)=\mathcal{C}(G,T_{e\to\bf{e}},a)$ for all $a\notin \{\bf{e}\}\cup\wp_{\bf{e}}$.
\end{prop}
\begin{proof}
Since $T$ is a spanning tree, adding the edge $\textbf{e}\notin E_T$ will produce exactly one cycle formed by $\bf{e}$ and the unique path $\wp_{\bf{e}}$ joining its extreme vertices. Removing any edge in this unique cycle will produce a subgraph without cycles and still connected, i.e., a spanning tree.

The connected components, of both $T_{e\to\bf{e}}$ and $T$, obtained by removing an edge which is not in the cycle $\{\bf{e}\}\cup\wp_{\bf{e}}$, are the same. Henceforth, the edge congestions of edges which are not in the cycle must be also the same (see Figure~\ref{f:cycle_perturbation}).

\end{proof}

\begin{figure}[h]
\centering
\includegraphics[scale=0.2]{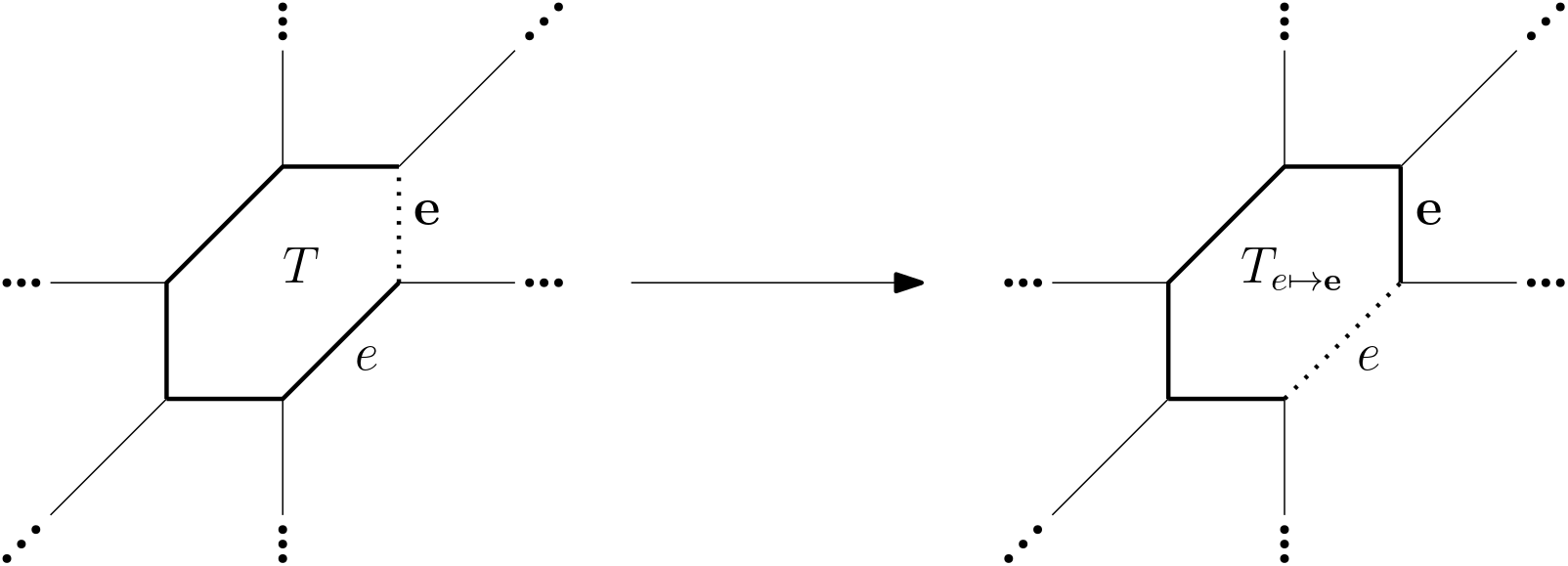}

\caption{Example of a cycle perturbation in a spanning tree. Observe that only the edge congestion of bold edges can change.}
\label{f:cycle_perturbation}
\end{figure}

From Proposition~\ref{p:equal_cong}, it follows that we only need to compute the edge congestions of edges of $T_{e\to\bf{e}}$ that belong to $\{\textbf{e}\}\cup\wp_{\bf{e}}$ in order to compute the  $L^p$-congestion of $T_{e\to\bf{e}}$. Usually, $|\wp_{\bf{e}}|$ is much lower than the number of vertices.

If for some $\textbf{e}\in E_T$ the $L^p$-congestion of some $T_{e\to\bf{e}}$ is lower than the original one, then we choose the best perturbed spanning tree $T_{e\to\bf{e}}$ for $e\in \wp$. We repeat this procedure until no perturbed tree has lower congestion. The given tree is a local minimum in $\cS_G$ for the spanning tree congestion problem

The congestion descent algorithm has the following work-flow (see also Algorithm~\ref{alg:sCD}).
\begin{itemize}
\item[Step 1] Given $T$, compute a list $c$ formed by the edge congestions of the edges of $T$ and a list of paths $\wp_\textbf{e}$ joining the extremes of each $\textbf{e}\in E_G\setminus E_T$ (complexity at most $O(m\cdot n)$).

\item[Step 2] For each $e$ in $\wp_{\bf{e}}$, compute $\cC(G,T_{e\to\bf{e}},a)$ for $a\in \{\textbf{e}\}\cup \wp_{\textbf{e}}\setminus\{e\}$ (complexity at most $O(m n^2)$).

\item[Step 3] For each edge $e$ in $\wp$, compute $\|(\cC(G,T_{e\to\bf{e}},a))_{a\in \wp_e}\|_p$. If the minimum of these norms is lower than the former $\|(\cC(G,T,e))_{e\in \wp_{\textbf{e}}}\|_p$, then assign $T \gets T_{e\to\bf{e}}$ for the minimal congestion edge $e\in \wp_{\textbf{e}}$ and repeat the process (this is $O(p)$ from the previous step).

\item[Step 4] If no perturbation is possible for evey $\textbf{e}\notin E_{T}$, then $T$ is a local minimum for the $L^p$-STC. Otherwise change $T$ by $T_{e\to\bf{e}}$ and repeat the process.
\end{itemize}

\begin{algorithm}
\caption{$\rm{sCD}(G,T,p)$, $G$ graph, $T$ spanning tree, $p\in\N\cup\{\infty\}$}\label{alg:sCD}

$c$: array, $\rm{len}(c)= V-1$\;
\For{$e\in T$}
{
	$c[e] \gets \cC(G,T,e)$\;
}
local\_minimum $\gets$ False\;
\While{\rm{local\_minimum} $=$ \rm{False}}
{
G $\gets$ Randomized $G$\;
local\_control $\gets$ False\; 
\For{$\textbf{e}\notin T$}
{
cycle\_path = Find\_Cycle($T$, $\bf{e}$)\;
local\_congestion: array, $\rm{len}(\text{local\_congestion}) = \rm{len}(\rm{cycle\_path})$\;
\For{$e\in$ \rm{cycle\_path}}
{
local\_congestion[e] $\gets$ c[e]\;
}
current\_congestion $\gets$ local\_congestion\;
\For{$e \in$ \rm{cycle\_path}}
{
new\_T $\gets T$\;
new\_T[$e$]  $\gets \bf{e}$\;
new\_cycle\_path $\gets$ cycle\_path\;
new\_cycle\_path[$e$] $\gets \bf{e}$\;			
\For{$a\in$ \rm{new\_cycle\_path}}
{
\eIf{$a\in$ \rm{cycle\_path}}
{
current\_congestion[a] $\gets \cC(G,T_e,a)$\;
}{
current\_congestion[$e$] $\gets \cC(G,T_e,\bf{e})$\;
}

}
\If{$\|\rm{current\_congestion}\|_p < \|\rm{local\_congestion}\|_p$}
{
local\_control $\gets$ True\;
$T\gets \rm{new}\_T$\;
local\_congestion $\gets$ current\_congestion\;
\For{$a\in$ \rm{new\_cycle\_path}}
{
\eIf{$a\neq \bf{e}$}
{
$c[a] \gets$ current\_congestion[a]\;
}{
$c[e] \gets$ current\_congestion[$\bf{e}$]\;
}
}

}
}
}
\If{\rm{local\_control} = \rm{False}}
{
local\_minimum = True\;
}

}
\Return{ $\|c\|_p$, $\rm{new}\_T$}

\end{algorithm}

The complexity of a single perturbation instance is at most $O(m\cdot n^2)$. It is not difficult to bound the maximum number of instances: observe that there is a finite number of edge weights, therefore, if a perturbation is made, the difference between the $L^p$ congestions is uniformly lower bounded by a number $\varepsilon >0$. Assume that $\varepsilon > 0$ is given, it follows that at most $\mathcal{C}_p(G,T)/\varepsilon$ perturbations can be done.

Observe that a rough upper bound for every $L^p$-congestion is the number:
$$
b_{p,\omega} = M_\omega\cdot m\cdot(n-1)^{1/p}\;,
$$
where $M_\omega = \max_{e\in E_G} \omega(e)$.

In the specific case where $\omega$ is constant $1$, i.e. an unweighted graph, it follows that $b_{\infty,\omega} = m$ and $b_{1,\omega} = m\cdot(n-1)$. 

Assume first that $\omega$ takes values in $\N$, the $L^\infty$ congestion of some edge is reduced by at least $1$ when a perturbation instance succeeds and the global congestion does not increase. It follows that at most $(n-1)b_{p,\omega}$ perturbation instances can occur. The previous estimate tells us that the complexity of the congestion descent algorithm  for $p=\infty$, for integer weighted graphs, is upper bounded by $O(M_\omega m^2\cdot n^3)$. In the specific case of an unweighted graph the complexity is upper bounded by $O(m^2n^3)$.

When $\omega$ is not integer valuated, we can multiply the weight function by a sufficienlty large constant $K$ such that the difference between every different weight is more than $1$. By the same reasoning as in the  integer weighted case, the complexity of the algorithm will be polynomial for $K\cdot\omega$. Since $\cC(G,T,e,K\omega)=K\cdot \cC(G,T,e,\omega)\,,$  it follows that the algorithm executed for $K\cdot\omega$ visits exactly the same trees as the algorithm executed for $\omega$ and the polynomial complexity also follows.

The cases $1\leq p<\infty$ can be treated in a similar way. Recall that minimizing $\cC_p(G,T)$ is equivalent to minimize $\cC_p(G,T)^p$. At each descent congestion instance $\cC(G,T)^p$ must decrease by an integer number if the weights are integers (otherwise we can use a constant $K$ as before). In this case $\cC(G,T)^p\leq b_{p_\omega}^p = M_\omega^p\cdot m^p\cdot (n-1)$, thus the complexity\footnote{Our computer experiments suggest that the algorithm is, in general much faster than these estimatives in typical graphs.} of the algorithm is, at most, $O(K^pM_\omega^p m^{p+1}n^3)$.

At each descent instance of the algorithm, we shall randomize the list of edges in $G$, the purpose of this randomization is to allow the algorithm to explore different parts of the graph at each step avoiding any choice bias. The previous algorithm will be called the {\em strict congestion descent algorithm} and will be noted as ${\rm sCD}_p$. 

\begin{rem}\label{r:deep_sCD}
Observe that the ${\rm sCD}_p$ algorithm for large $p$ can be used to obtain interesting estimations of the $L^\infty$-STC by means of Proposition~\ref{p:Holder}. This approach has demonstrated to be relevant for some graphs as hypercube graphs. This version of the algorithm, i.e., to get the minimum $L^q$-congestion spanning tree among those explored by the ${\rm sCD}_p$ algorithm, will be denoted by ${\rm sCD}_{p,q}$. In this paper we only test the ${\rm sCD}_{10,\infty}$ as a proof of concept.
\end{rem}

The performance and discussion of the algorithm ${\rm sCD}_p$ in several families of graphs is given in the last section.  We also recall that we do not claim that our implementation is the best possible, significative gains in the theoretical complexity is expected in the future.

\section{Planar graphs}

\begin{defn}
A graph $G=(V,E)$ is called {\em planar} if there exists inclusions $i_V:V\to \R^2$ and $i_E:E\to C^\infty([0,1],\R^2)$ such that 
\begin{enumerate}
\item $i_E(e)(0)$ and $i_E(e)(1)$ are the images by $i_V$ of the adjacent vertices of $e$ for all $e\in E$.

\item $i_E(e)([0,1])\cap i_E(e')([0,1]) = i_V(e\cap e')$.

\item Each $i_E(e)$ is an embedding for all $e\in E$.
\end{enumerate}
Given a planar graph, we will always assume that the emebedding is implicit (i.e. both $i_V$ and $i_E$ were already defined). In this context, we will assume that $V\subset \R^2$ and each $e\in E$ is a smooth path which connects its adjacent vertices.

The trace of a path $e\in E$ will be denoted by $\bm{e}$ (in bold letter). Let $G=(V;E)$ be a planar graph, the subset of $\R^2$ consisting of the vertices and the traces of the edges will be denoted by ${\bm G}$ .
\end{defn}

\begin{defn}
A {\em cell} of a planar graph $G$ is the closure of a  connected component of $\R^2\setminus {\bf G}$.

\end{defn}

\begin{defn}[Dual Graph]
Let $G=(V,E)$ be a connected planar graph and let $V^\ast$ be the set of cells of $G$. The {\em dual graph} of $G$, denoted by $G^\ast$, is the multigraph $(V^\ast,E^\ast)$ such that $\overline{PQ}\in E^\ast$ if and only if there exists an edge in $E_G$ which is adjacent to the cell $P$ on one of its sides and to the cell $Q$ in its other side.

If $(G,\omega)$ is a weighted graph, then we can define a {\em dual weight} $\omega^\ast$ on $G^\ast$:  observe that every edge $e\in E_G$ defines a unique edge $e^\ast$ in $E^\ast$, define $\omega^\ast(e^\ast)=\omega(e)$. The edge $e^\ast$ will be called the {\em dual edge} of $e$.
\end{defn}

\begin{rem}\label{r:planar dual}
The dual graph is also a planar (multi)graph. Just place a vertex in the interior of each cell (called {\em centroids}) and each edge $e^\ast$ of the dual graph is represented as a smooth path joinig the respective centroids and meeting transversely a unique edge $e$ in $G$, exactly the one with same weight as the dual weight of $e^\ast$. This also  allows to define the dual edge $e^\ast\in E^\ast$ for each edge $e\in E$, in particular both graphs have the same number of edges.

Remark also, that the dual graph of a planar graph is always connected since every cell is connected with the unbounded cell by a path formed by finitely many dual edges.
\end{rem}

\begin{prop}[Folklore]
Let $G$ be a, possibly edge-weighted, planar graph. Let us consider the representation of $G^\ast$ in $\R^2$ given in Remark~\ref{r:planar dual}, then $(G^\ast)^\ast$ is isomorphic to $G$ and $(\omega^\ast)^\ast = \omega$.
\end{prop}

\begin{defn}[Dual Tree]
Let $G$ be a connected planar graph. Let $T$ be a spanning tree. Define the {\em dual tree} of $T$, denoted by $\widetilde{T}$, as the subgraph of $G^\ast$ generated by the edges of $G^\ast$ which are {\bf not} dual to edges of $T$.
\end{defn}

\begin{prop}
Let $G=(V,E)$ be a connected planar graph. Let $T$ be a spanning tree of $G^\ast$. Then $\widetilde{T}$ is also a spanning tree of $G$.
\end{prop}

\begin{proof}
On the one hand $\widetilde{T}$ cannot have a cycle, otherwise that cycle will disconnect $\R^2$ in several connected components. At least one cell will intersect one bounded connected component and at least one cell will intersect the unbounded one, since $T$ is spanning and connected it follows that there exists a path in $T$ connecting those cells. This path must meet transversely the cycle at some of its edges and that edge is dual to some edge in $T$. This is in contradiction with the definition of dual tree.

On the other hand, $\widetilde{T}$ must be spanning. If this were not the case, then $T$ must have at least one cycle enclosing a vertex in $V_G$ but this is impossible since $T$ is a tree. By a similar reason $\widetilde{T}$ must be connected, otherwise a cycle in $T$ would disconnect  $\widetilde{T}$ in (at least) two components.
\end{proof}

\begin{cor}
Let $G$ be a connected planar graph. Then $T$ is a spanning tree of $G$ if and only if $\widetilde{T}$ is a spanning tree of $G^\ast$. 
\end{cor}

The edge congestion of a spanning tree of a reduced planar graph is easy to compute from its dual tree. The next Proposition is a generalization of \cite[Pag. 6 (Key Observation)]{Law-Ostrovskii}.

\begin{prop}[Edge Congestion in a planar graph]\label{p:edge planar congestion}
Let $(G,\omega)$ be an edge-weighted connected planar graph and let $T$ be a spanning tree. Let $\widetilde{T}$ be its dual tree in $(G^\ast,\omega^\ast)$. Let $e\in E_T$ and let $e^\ast =\overline{PQ}$ be its dual edge. Let $\wp_{e^\ast}$ be the shortest $(P,Q)$-path in $\widetilde{T}$. It follows that $\mathcal{C}(G,T,e)=w(e) + w^\ast(\wp_{e^\ast})$, where $\omega^\ast(\wp_{e^\ast})$ denotes the sum of the dual weights of the (dual) edges in $\wp_{e\ast}$.
\end{prop}

\begin{proof}
Let $e\in E_T$ and let ${\wp_{e^\ast}}$ be the shortest path in $\widetilde{T}$ joining the cells in $e^\ast$. Since $G^\ast$ is also planar, it follows that $\wp_{e^\ast}$ can be embedded in $\R^2$ and also $e^\ast$ can be represented in $\R^2$ by the planar embedding depicted in Remark~\ref{r:planar dual}. Let $T_e$ denote the graph obtained from $T$ by removing the edge $e$. It is clear that ${\bm\wp}_{e^\ast}\cup {\bf e^\ast}$ is the trace of a simple loop in $\R^2$ that separates both connected components of $T_e$. The dual weights of the edges of this cycle are, by definition, the corresponding weights of the edges in $G$ that connect both components (recall that each edge meets transversely just one dual edge) and therefore the sum of these weights are precisely the edge congestion of $e$.

In the case where $P=Q$ (for instance when $e$ is adjacent to a degree $1$ vertex of $G$), $\wp_{e^\ast}$ is defined as an empty path (with length $0$).
\end{proof}

From the previous Proposition we obtain an easy lower bound for the edge congestion in edge-weighted planar graphs. The next is a generalization of the cycle rank of a graph as defined in \cite[Pag. 8]{Law-Ostrovskii}

\begin{defn}
Let $G$ be a connected planar graph and let $G^\ast$ be its dual graph. Let $e^\ast$ and edge of $G^\ast$ which is dual to an edge $e\in E_G$. Let $c_e^\ast$ denote a cycle in $G^\ast$ such that it contains $e^\ast$ and its dual weight is minimal.

The number $w^\ast(c_e^\ast)$ is called the {\em weighted cycle rank} of $e$ and it will be denoted by $\rm{r}_\omega(e)$.
\end{defn}

\begin{cor}
Let $(G,\omega)$ be an edge-weighted connected planar graph. Then $\rm{r}_\omega(e)\leq\mathcal{C}(G,T,e)$. In particular, $\max_{e\in E_G}\{\rm{r}_\omega(e)\}\leq  \mathcal{C}_\infty(G,\omega)$.
\end{cor}

\begin{cor}
Let $T_{\rm{R}}$ be a minimal spanning tree of the edge-weighted graph $(G,\rm{r}_\omega)$. Then

$$\left(\sum_{e\in E_{T_{\rm{R}}}}(\textrm{r}_\omega(e))^p\right)^{1/p}\leq  \mathcal{C}_p(G,\omega)\;.$$
\end{cor}
\begin{proof}
Let $T$ be a spanning tree of $G$. Let $(e_i)_{i=1}^k$ be an ordering of elements in $T$ by weighted-rank in increasing order and let $(e'_i)_{i=1}^k$ be a similar ordering for the edges in $T_{\rm{R}}$. By the Kruskal algorithm (see the proof of Proposition~\ref{p:m bound}), it follows that $\textrm{r}_\omega(e'_i) \leq \textrm{r}_\omega(e_i)$ for all $1\leq i\leq k$ and therefore 

\begin{align*}
\left(\sum_{e\in E_{T_{\rm{R}}}}(\textrm{r}_\omega(e))^p\right)^{1/p} &=\|(\textrm{r}_\omega(e'_i))_{i=1}^k\|_p\\ &\leq \|(\textrm{r}_\omega(e_i))_{i=1}^k\|_p \leq  \|(\mathcal{C}(G,T,e_i))_{i=1}^k\|_p = \mathcal{C}_p(G,T,\omega)
\end{align*}
\end{proof}

From the Proposition~\ref{p:edge planar congestion}, a good candidate for a minimal $L^p$-congestion spanning tree in an (unweighted) planar graph is a tree whose dual tree minimizes the stretch in the dual graph; this was already observed for the classical congestion in \cite{Law-Leung_Ostrovskii2014}. In the case of weighted graphs, we have to minimize the weights of paths but the idea remains the same.

\begin{rem}\label{r:radial_dual}
If $G^\ast$ admits a radial spanning tree $T$, i.e. a tree where a root vertex is connected with everything, then it is clear that $T$ minimizes the lengths of paths for every possible pair of cells in the dual graph. In this specific case $\widetilde{T}$ has minimal $L^p$-congestion for every $1\leq p\leq \infty$ whenever $G$ is unweighted, the edge congestion of every edge varies between $1$ and $3$.
\end{rem}

Weighted planar graphs whose dual graphs admit a radial spanning tree are the simplest ones for the study of the spanning tree congestion problem. However, even in this case, it is not true in general that the dual tree of a minimum $L^p$-congestion spanning tree is a radial spanning tree as it can be seen in the following example.

\begin{exmp}\label{ex:nonradial}
Let $G$ be the weighted\footnote{Observe that these weights satisfy the triangle inequality.} graph presented in Figure~\ref{f:nonradial}. Its dual graph $G^\ast$ admits a unique radial spanning tree. That one is dual of another spanning tree of $G$ whose $L^1$-congestion is $50$ (left). The spanning tree depicted in the right has $L^1$-congestion $42$, observe that its dual tree is linear (far from radial).
\begin{figure}[h]
\centering
\includegraphics[scale = 0.22]{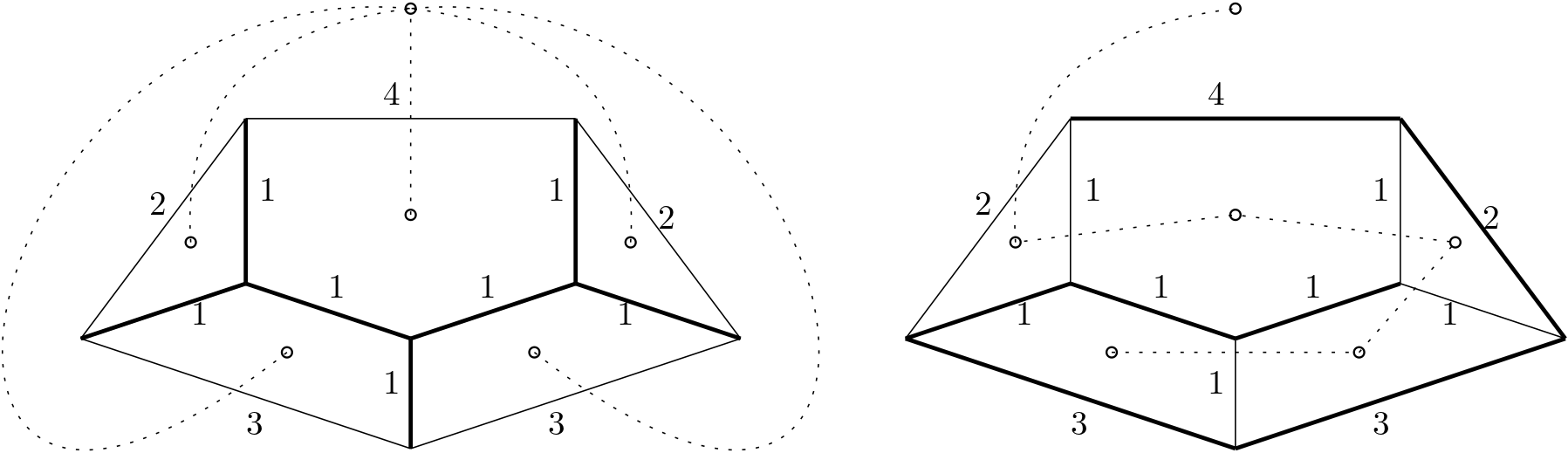}
\caption{Radial dual trees do not need to be optimal for weighted planar graphs.}
\label{f:nonradial}
\end{figure}
\end{exmp}

However, there is a family of planar weighted graphs that always admit a radial minimal $L^p$-congestion spanning tree.

\begin{defn}
A connected graph $G$ is called  {\em cactus graph} if and only every edge belongs to at most one cycle.
\end{defn}

It is easy to show that every cactus graph is planar. Moreover, a planar graph is a cactus graph if and only if the maximal simple subgraph of the dual graph is a radial spanning tree.

\begin{prop}\label{p:minimal-cactus}
Let $(G,\omega)$ be a weighted cactus graph embedded in $\R^2$. Let $T$ be a minimal weight spanning tree in $(G^\ast,\omega^\ast)$. Then $\widetilde{T}$ is a minimal $L^p$-congestion spanning tree for every $p\in\N\cup\{\infty\}$.
\end{prop}
\begin{proof}
Every spanning tree of the dual graph of a cactus graph has the same root cell, namely $P$ and must be radial. Every cycle in the dual graph of a cactus graph has length $2$ (or $1$ if $G$ has a loop). By means of Proposition~\ref{p:edge planar congestion}, it follows that, for every spanning tree $\widetilde{T}$ of $G$, we get $\mathcal{C}(G,\widetilde{T},e)=w(e) + w^\ast(\overline{e^\ast})$ where $\overline{e^\ast}$ is the dual edge in the dual tree with the same adjacent vertices as $e^\ast$. Thus, the minimal weight spanning tree in the dual graph (with dual weights) minimizes each coordinate in the vector of edge-congestions and henceforth its dual minimizes the spanning tree $L^p$-congestion.
\end{proof}

Remark~\ref{r:radial_dual} and Propositions~\ref{p:edge planar congestion} and \ref{p:minimal-cactus} suggest that good candidates to minimal congestion spanning trees in planar graphs are those whose dual trees minimizes their weights and maximizes the degree at the minimum number of vertices (trying to be as radial as possible).

In the absence of weights, a first approach to this question are breadth first search trees or BFS-trees.

\begin{defn}[BFS Tree]
A spanning tree of a connected edge-weighted graph is called of {\em breadth first search} tree if and only if there exists a vertex $v$ such that the distance of any vertex to $v$ in the tree is the same as the distance in the graph. The vertex $v$ will be called the {\em root} of the BFS tree.
\end{defn}

In \cite[Proposition~1]{Ostrovskii2010}, BFS trees rooted at the unbounded cell are used to give an upper estimate to the $L^\infty$-congestion of connected planar graphs. However, we can consider arbitrary BFS trees in order to enhance that upper bound. Moreover, the example given in \cite[Figure~2]{Ostrovskii2010} is, in fact, the dual of a BFS tree rooted at one of the bounded cells. This is equivalent to change the planar embedding of the graph.

Usually, the STC on planar unweighted graphs is attained by a spanning tree whose dual is a BFS tree (this happens in lots of the examples studied in this work). We do not think that this is the general picture, but we still have no counterexample to this conjecture.

Thus, for unweighted graphs, a reasonably good upper bound for the $L^p$ congestion is given by the minimal $L^p$-congestion among the family of spanning trees which are dual to BFS trees in $G^\ast$. However, the set of BFS trees in the dual graph is, in general, very high to be computationally resolved.

\subsection{Second Algorithm: Branch optimization in BFS trees}

\begin{defn}
Let $T$ be a tree and $v\in V_T$ a distinguished vertex (called {\em root}). The {\em level} of a vertex $w\in V_T$ is its distance to the root. The $k$-level of $T$ is the set of vertices whose level is $k$, denote this set by $\Lvl_k(T,v)$.
\end{defn}

\begin{defn}[Relative $L^p$ congestion]
Let $G$ be a simple planar graph, let $T$ be a tree of the dual graph $G^\ast$ ($T$ does not need to be spanning).
Let $E^T\subset E_G$ be the set of edges $e$ such that $e^\ast$ is not an edge in $T$ and $e^\ast$ is adjacent to vertices of $T$. Let $\widehat{T}$ be any spanning tree in $G$ whose edge set includes $E^T$. Observe that the $L^p$-congestion $\mathcal{C}(G,\widehat{T},e)$ of every $e\in E^T$ can be computed by the formula given in Proposition~\ref{p:edge planar congestion} and does not depend on the choice of $\widehat{T}$.

The Relative $L^p$ congestion of $T$, denoted by  $\mathcal{RC}_p(G,T)$ is defined as 
$$
\max_{e\in E^T}(\mathcal{C}(G,\widehat{T},e))\,,\ \text{if}\ p=\infty
$$

$$
\dfrac{\sum_{e\in E^T}\mathcal{C}(G,\widehat{T},e)^p}{| E^T|}\,,\ \text{if}\ 1\leq p<\infty
$$

\end{defn}

Of course, when $T$ is a spanning tree and $1\leq p <\infty$, we get $\widehat{T}=\widetilde{T}$ and $$\left((|V_G| -1)\cdot\mathcal{RC}_p(G,T)\right)^{1/p} = \mathcal{C}_p(G,\widetilde{T})\;.$$

\begin{rem}\label{r:fast-planar-congestion}
The edge-congestions of edges in $E^T$ for a tree in $G^\ast$ can be computed very fast. By means of Proposition~\ref{p:edge planar congestion}, each edge congestion depends on the computation of a path in the tree $T$, this can be done in at most $O(|V_T|)$. Therefore, the computation of the relative congestion of $T$ can be done in, at most, $O(|E^T|\cdot|V_T|)$ time.
\end{rem}

\begin{defn}
Let $T$ be a rooted spanning tree of a graph. Let $T_k$ be the subtree of $T$ generated by the edges that join vertices in levels $\leq k$.

Let $T$ and $T'$ be BFS spanning trees with the same root. It is said that $T'$ is a {\em switch} at level $k$ from $T$ if $T'_k$ and $T_k$ are the same up to remove just one edge (see Figure~\ref{f:switch}).
\end{defn}

\begin{figure}[h]
\centering
\includegraphics[scale=0.2]{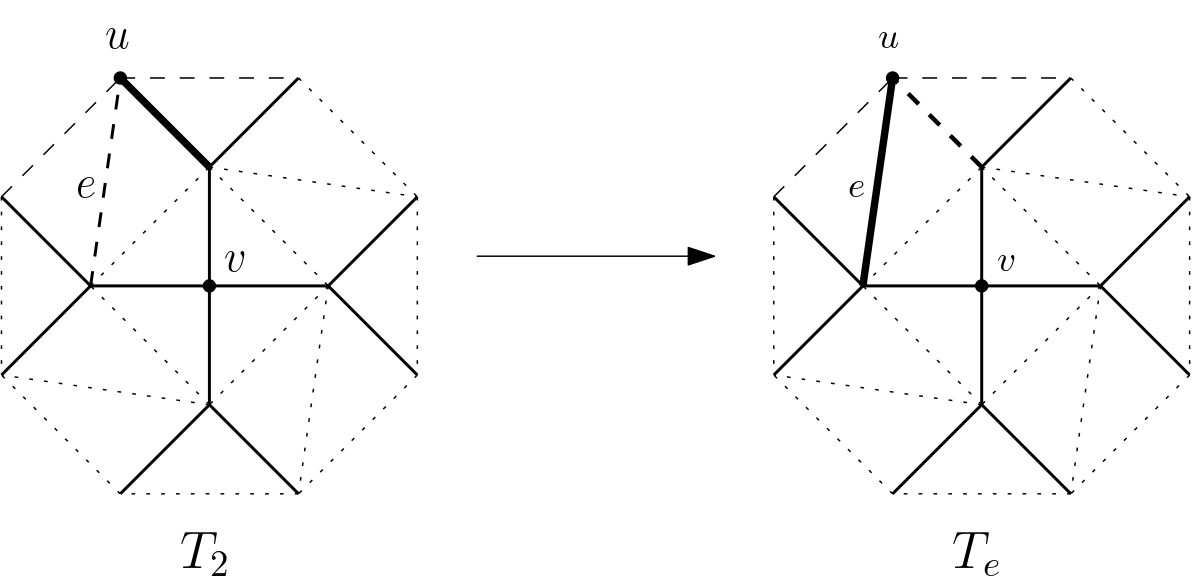}
\caption{Example of a switch at level $2$ in a BFS tree. Bold lines represent the tree at level $2$, dotted and dashed edges are those in $E^{T_2}$, and $T_e$ is obtained by switching the unique edge of $T_2$ adjacent to $u$ with $e$. If the tree lies in the dual graph, then dotted and dashed edges are dual to edges where the congestion can be computed. From $T_2$ to $T_e$ only the edge-congestion of the (duals of) dashed segments, which are those incident in $u$, may change.}
\label{f:switch}
\end{figure}

\begin{defn}[Locally Optimal Congestion BFS tree]
Let $G$ be a simple planar graph, let $T$ be a BFS spanning tree of the dual graph $G^\ast$ and let $v\in V^\ast$ be the root cell. It is said that $T$ has {\em Locally Optimal $L^p$-Congestion} ($p$-LOC tree) if and only if, for every $k\geq 0$, 
$$\mathcal{RC}_p(G,T_{k})\leq \mathcal{RC}_p(G,T'_{k}) $$ for every tree $T'$ which is a switch at level $k$ from $T$.
\end{defn}

Our second algorithm produces $p$-LOC-BFS spanning trees for each possible root cell in the dual graph in polynomial time. The upper bound in the $L^p$ congestion is chosen among the minimal $L^p$-congestions of the respective dual trees. For the implementation of this algorithm it is assumed that edges are represented as line segments in $\R^2$. Observe also that here $n = |V_{G^\ast}|$, $m=|E_{G^\ast}| = |E_{G}|$ and $\Delta$ is the maximum degree of $G^\ast$.

The algorithm follows the next workflow (see Algorithm~\ref{alg:LOC-BFS}):

\begin{itemize}
\item[Step $1$:] Choose $v$ as the root cell, define $T_1$ as the family of dual edges incident to the chosen root. Let $S_1$ be the vertices in $T_1$. Let $E_1$ be the edges in $E_{G^\ast}\setminus E_{T_1}$ adjacent to vertices in $S_1$. Compute the relative congestion at level $1$: $\mathcal{RC}_p(G,\tilde{T}_1,1)$. 

\item[Step $2$:] Inductively, assume $T_k$, $S_k$ and $E_k$ were already defined and the relative congestions were computed up to level $k$. Define $B_{k+1}$ as those cells in $V_{G^\ast}\setminus S_k$ that can be joined with $S_k$ by an edge, set $S_{k+1}=S_k\cup B_{k+1}$. Let $OE_{k+1}$ be the set of edges adjacent to vertices in $S_{k+1}$.

\item[Step $3$:] Let $T_{k+1}$ be an arbitrary BFS extension of $T_k$ to the vertex set $S_{k+1}$ and compute its relative $L^p$-congestion.

\item[Step $4$:] For each $u\in B_{k+1}$,  let $BE^u_k$ be the set of edges in $E_k$ adjacent to $u$. If $|BE^u_k|>1$ for some $u\in B_{k+1}$, then for each $e\in BE^u_k$ compute the relative congestion of the tree $T_e$ obtained by switching the edge $e$ with the edge in $T_{k+1}$ adjacent to $u$. Observe that the unique dual edges where the congestion can vary are those in $BE^u_k\setminus\{e\}$ and those in $OE_{k+1}$ incident in $u$. If $ \mathcal{RC}_p(G,\tilde{T}_{e},k+1)<\mathcal{RC}_p(G,\tilde{T}_{k+1},k+1)$, then replace $T_{k+1}$ with $T_e$, and repeat the process until $T_{k+1}$ is a locally optimal tree for switches at any level $j\leq k+1$.

\item[Step $5$:] Repeat the Step $4$ until $T_{k+1}$ is a spanning tree.

\end{itemize}

\begin{algorithm}

\caption{$\rm{LOC\_BFS}(G,V,p)$, $G$ planar graph, $V\subset\R^2$ vertex list, $p\in\N\cup\{\infty\}$}\label{alg:LOC-BFS}
\algsetup{linenosize=\small}
  \scriptsize
$T$: array, $\rm{len}(T)= 0$\;
$c$: array, $\rm{len}(c)= 0$\Comment*[r]{array of edge congestions}
$\tilde{G}$, $\tilde{V} \gets$ Dual\_Graph($G$,$V$) \Comment*[r]{$ \tilde{V}$: array, vertices(cells) of $\tilde{G}$}
\For{$v\in \tilde{V}$}
{
	$S\gets [v]$ \;
	$B\gets [v]$\;
	\While{$\rm{len}(T) < \rm{len}(V) -1$}
	{		
		P $\gets$ Edges($\tilde{G}$,$S$, $V\setminus S$)\;
		New\_$S \gets$ $S\cup\rm{Vertices}(P)$\;
		New\_$B \gets \text{New}\_S \setminus S$\;
		External\_Dual = Edges($\rm{New\_}B$,$\rm{New\_}B$)\;
		Edges : double array, \rm{len}(Switch\_Edges) = 	\rm{len}(New\_$B$)\;	
		\For{$p\in \rm{New\_}B$}
		{
		Edges[$p$] $\gets$ Edges(p,B)\;
		$T$\text{.append}(Edges[$p$][$0$])
		}
		$B \gets \rm{New\_B}$\;
		\For{$e\in$ \rm{External\_Dual}}
		{
		$\rm{c.append}(\cC(G,\rm{Dual}(T),\tilde{e})) $
		}
		\For{$p\in \rm{New\_}B$}
		{
		\If{$\rm{len}(Edges[p])>1$} 
			{	
			
			\For{$i\in \{1,\dots,\rm{len}(Edges[p])-1\}$}
				{
					$c\rm{.append}(\cC(G,\rm{Dual}(T),\widetilde{\text{Edges}[p][i]}))$ \Comment*[r]{$p$ is a switching point}
				}
			}
		}
		local\_minimum $\gets$ False\;
		\While{\rm{local\_minimum} $=$ \rm{False}}
		{
		Aux\_$T \gets T$\;
		Aux\_$c \gets c$\;
		local\_control = False\;
		\For{$p\in\rm{New\_}B$}
		{
			edge $= T	\cap \rm{Edges}[p]$	\;
			\If{$\rm{len}(\rm{Edges}[p]) >1$}
			{
			\For{$e\in \rm{Edges}[p]$, $e\neq\rm{edge}$}
			{
			$\rm{Aux\_}T[\rm{edge}] \gets e$\;
			$\rm{Aux\_}c[e] \gets \cC(G,\rm{Dual}(\rm{Aux\_}T),\widetilde{\rm{edge}})$\;
			\For{$a\in \rm{Edges}[p]$}
			{
				\If{$a\neq e$}
				{
				$\rm{Aux_c}[a]\gets \cC(G,\rm{Dual}(T),\widetilde{a})$\; 
				} 
			}
			\For{$a\in \rm{External\_Edges}$}
			{
				\If{$p\in a$}
				{
				$\rm{Aux\_}c[a] \gets \cC(G,\rm{Dual}(\rm{Aux\_}T),\widetilde{a})$ 
				}
			}
			\If{$\|\rm{Aux\_}c\|_p <\|c\|_p$}
			{
			local\_control = True\;
			$T \gets \rm{Aux\_}T$\;
			$c \gets \rm{Aux\_}c$\;
			\rm{\textbf{break for: line 35}}\;
			\rm{\textbf{break for: line 32}}\; 
			}
			}
			
			}		
		}
		\If{$\rm{local\_control} = \rm{False}$}
		{
		local\_minimum = True\;
		}
		}

	}
}

\Return{ $\|c\|_p$, $\rm{Dual}(T)$}
\end{algorithm}

This algorithm will be denoted by ${\rm LOC\textendash BFS}_p$ and
it is clear that it runs in polynomial time. The computation of the $S_k$'s, $B_k$'s, $BE^u_{k}$'s and $OE_k$'s run also in linear time with the number of edges and can be done at the very beginning of the algorithm. Since this data is previously known, the vector of edge congestions of a tree $T_k$ can be obtained in $O(\Delta|B_k|)$ time since, at each step, the paths computed for $T_{k-1}$ can be used again to make a faster computation of the paths in $T_k$ and there are at most $\Delta\cdot |B_{k}|$ edges where the congestion must be computed. Observe that we  know, a priori, the extreme points of relevant paths to be considered for edge-congestion computation. 

The fifth step also runs in polynomial time for any initial random extension $T_{k+1}$ (and assuming the weight function is fixed). A switch operation can be performed in $O(\Delta)$ time. The reason for this is the fact that we need to compute a fairly small number of edge congestions at each switch operation: only congestions of edges which are dual to edges in $(BE^u_k\setminus\{e\})\cap E^{T_e}$ and edges adjacent to $u$ in $OE_{k+1}$ can change in a switch operation, the number of such edges is clearly bounded by the maximum vertex degree $\Delta$.

In the case $p=\infty$, the relative congestion of $T_{k}$ is, by construction, clearly bounded by $(2k+1)M_\omega$ where $M_\omega = \max_e\{\omega(e)\}$ since $2k$ is the diameter of $T_k$. If $\omega$ is integer valuated, at each succeeded switch operation, the relative congestion must descend by at least $1$. Therefore, a locally optimal tree (up to level $k$) is reached in, at most, $O(\Delta^2 \cdot |B_k|\cdot k\cdot M_\omega)$ operations. Since the sets $B_{k}$'s form a partition of $|V|$ and $k$ cannot be larger than twice the diameter of $G^\ast$. It follows that the global complexity for finding the final $\infty$-LOC-BFS tree is, therefore, at most $O(\Delta^2\cdot n\cdot d\cdot M_\omega)$, where $d$ is the diameter of $G^\ast$. The global complexity of the algorithm (running at every possible rooted vertex) is, at most, $O(M_\omega\Delta^2\cdot n^2d)$. The real valuated weight case and $p\in\N$ can be handled in a similar way as it was done in the first sCD algorithm. The complexity can, in general, increase depending in the weight function and the value of $p$.

\subsection{Third Algorithm: Recursive Optimal Congestion  in the dual graph }

For weighted spanning trees, BFS trees are not longer good options. The dual tree of a minimal $L^p$-congestion spanning tree of $G$ tends to avoid dual edges with large weight; otherwise, this large weight will compute at least twice for edge congestion computations (this is of great relevance when $p<\infty$). Therefore minimal $L^p$-congestion spanning trees on weighted graphs have a preference to include the heaviest edges.

For this third algorithm, we adapt some of the strategies given in the ${\rm LOC\textendash BFS}_p$ algorithm without using BFS trees. We shall compare the performance of both algorithms in the last section. As before $n = |V_{G^\ast}|$ and $m=|E_{G^\ast}| = |E_{G}|$.

The algorithm, denoted by ${\rm ROC}_p$, is described by the following recursive process (see also Algorithm~\ref{alg:ROC}):

\begin{itemize}
\item[Step $1$:] Choose a root cell $v_0\in V_{G^\ast}$. Set $S_0=\{v_0\}$ and let $T_0$ be the subgraph consisting of a single vertex $v_0$.

\item[Step $2$:] Suppose that $S_k$ and $T_k$ were defined. Define $B_{k+1}$ as those vertices in $V\setminus S_k$ adjacent to some vertex in $S_k$. For each $u\in B_{k+1}$, let $SE_u$ be the set   of edges of $G^\ast$ from $u$ to $S_k$.

\item[Step $3$:] If $|SE_u|\leq 1$ for all $u\in B_{k+1}$, then choose an edge $e$ from $S_k$ to $B_{k+1}$ such that $\omega(e)$ is minimum. If there are several such minimal edges, then choose the one that is incident with the vertex $w$ in $S_k$ where $\deg(w,\omega)/\deg(w)$ is minimum. If there are still more than one vertex to choose then, choose the vertex with higher degree among them. If there are still more than one choice, then take the vertex in $S_k$ with lowest index. This criterion for choosing the minimal edge is encoded in the subprogram Minimal\_Edge. Define $S_{k+1}$ as the union of $S_k$ with the vertex in $B_{k+1}$ adjacent to $e$.

\item [Step $4$:] If $|SE_u|> 1$ for some $u\in S_{k+1}$, then for each such $u$ and for each $e\in SE_u$, let us consider the trees $T_e = T_k + e$. Set $T_{k+1}$ as the tree $T_e$ with minimum relative $L^p$-congestion. If there are several choices with the same relative congestion, then follow the same criterion as in the previous step in order to choose $e$ univocally. Set $S_{k+1}$ as the vertex set of $T_{k+1}$.

\item[Step $5$] Repeat the previous process until a spanning tree is produced.
\end{itemize}

\begin{algorithm}
\caption{$\rm{ROC}(G,V,p)$, $G$ planar graph, $V\subset\R^2$ vertex list, $p\in\N\cup\{\infty\}$}\label{alg:ROC}
\algsetup{linenosize=\small}
  \scriptsize
$T$: array, $\rm{len}(T)= 0$\;
$c$: array, $\rm{len}(c)= 0$\Comment*[r]{array of edge congestions}
$\tilde{G}$, $\tilde{V} \gets$ Dual\_Graph($G$,$V$) \Comment*[r]{$ \tilde{V}$: array, vertices(cells) of $\tilde{G}$}
\For{$q\in \tilde{G}$}
{
$S$.append($q$)\;
\While{$\rm{len}(T) <\rm{len}(\tilde{V}) -1$}
{
Boundary\_Edges = Edges($S$,$V\setminus S$)\;
$B$ = Vertices(Boundary\_Edges)$\setminus S$\;
$W$: double array, \rm{len}(W)= \rm{len(B)}\;
\For{$v\in B$}
{
$W[v] \gets$ Edges($\tilde{G}$, $S$, $v$)\;
}
Rel\_Congestion = $\infty$\;
Rel\_$c \gets c$\;
Rel\_$T \gets T$\; 
	\eIf{$\rm{len}(W[v]) = 1\,,\ \forall v\in B$}
  	{
  	edge $\gets$ Minimal\_Edge($\tilde{G}$,Boundary\_Edges, $S$)\; 	
  	$T$.append(edge)\;
  	$S$.append(Target(edge))\;
  	}
  	{
  	Edges $\gets \emptyset$\;
  	Rel\_List $\gets \emptyset$\;
  	\For{$v\in B$}
  		{	
  		\For{$e\in W[v]$}
  		{
  		\rm{Aux\_}$T \gets T\cup\{e\}$)\;
   		\rm{Auc\_}$c \gets c$\;
		\For{$a\in W[v]\,,\ a\neq e$}
			{
			\rm{Aux\_}$c$\rm{.append}($\cC(G,\rm{Dual}(\rm{Aux\_}T),a)$)
			}
		}
		\eIf{$p<\infty$}
		{
		\rm{Aux\_Rel\_Congestion} = $\dfrac{\|\rm{Aux\_}c\|_p^p}{\rm{len}(c)+\rm{len}(W[v])-1}$\;
		}
		{
		\rm{Aux\_Rel\_Congestion} = $\|\rm{Aux\_}c\|_\infty$
		}
		\If{$\rm{Aux\_Rel\_Congestion} < Rel\_Congestion$}
		{
		Edges $\gets \{e\}$\;
		Rel\_List $\gets \{\rm{Auc\_}c\}$\;
		\rm{Rel\_Congestion} $\gets$ \rm{Aux\_Rel\_Congestion}\;
		Rel\_$c \gets \rm{Aux\_}c$\;
		Rel\_$T \gets \rm{Aux\_}T$\;  		
		}
		\ElseIf{$\rm{Aux\_Rel\_Congestion} = Rel\_Congestion$}
		{
		Edges.append(e)\;
		Rel\_List.append(\rm{Auc\_}c)\;
		}
 		}
 	edge = Minimal\_Edge($\tilde{G}$,Edges,$S$)\;
  	$T \gets T\cup\{\rm{edge}\}$\;
  	$c$ = Rel\_List[edge]\;
  		}
  	}
}
\Return $\|c\|_p$, $\rm{Dual}(T)$
\end{algorithm}

The ${\rm ROC}_p$ algorithm also runs in polynomial time. Computing $S_k$ and the sets $SE_u$ for each $u\in B_{k+1}$  is done, inductively, in the second step in, at most, $O(\Delta)$ time (we only add a single vertex to the tree at each step). Step $3$ is for free as can be made as a side computation in the fourth step. Step $4$ depends on the computation of at most $\Delta\cdot |B_{k+1}|$ relative $L^p$-congestions of trees with exactly $k+1$ vertices. For each one of these trees we only need to compute $|BE_u|\leq\Delta$ edge congestions for some $u\in|B_{k+1}|$, in order to obtain the corresponding relative congestions. All these edge congestions can be computed with complexity $O(\Delta)$ since we can reuse the paths already obtained for the computation of the tree $T_{k}$. Hence, since $|B_{k+1}|\leq k\Delta$, the complexity of this step is, at most, $O( \Delta^2\cdot k)$.

Since at each application of steps $2$, $3$ and $4$, the tree adds a vertex, the number of instances of the algorithm is $n-1$. Observe that the Minimal\_Edge choice does not add complexity since we only need to get information of the degree, weighted degree and index of an already known list of vertices. Thus, the complexity of the algorithm at a chosen root vertex is at most $O(\Delta^2\cdot n^2)$ and the total complexity for applying the full workflow to every possible rooted vertex is at most $O(\Delta^2\cdot n^3)$.

\begin{rem}
Observe that the execution time of ${\rm ROC}_p$ is polynomial for every weight function and its complexity does not depend on $\omega$ or $p\in\N\cup\infty$. This is an important advantage respect to the ${\rm sCD}_p$ and ${\rm LOC\textendash BFS}_p$ algorithms.
\end{rem}

\section{Computer experiments}

In this section the performance of the algorithms is tested. Since the Congestion Descent algorithm is of general purpose whereas LOC-BFS and ROC algorithms only work for planar graphs. We consider a first subsection for nonplanar graphs, where only the strict Congestion Descent algorithm is tested, and a second subsection, where sCD is compared with LOC-BFS and ROC. In a third subsection we make similar estimations in some weighted graphs.

Almost all the considered families are unweighted graphs because of the lack of explicit computations for weighted ones (and, therefore, it is not clear whether the given results are good estimatives). We only consider in this work $p=\infty$ and $p=1$ versions of these algorithms. The notation $a(b)$ will make reference to the situation where the $L^\infty$-congestions of the spanning trees obtained with $\rm{sCD}_\infty$ and $\rm{sCD}_1$ are $a$ and $b$, respectively, with $b<a$; this notation is also used with $\rm{LOCBFS}_p$ and $\rm{ROC}_p$ algorithms but in this case, we also tested $p\in\{1,10,20,30,40,50\}$.

$L^\infty$-congestions are highlighted in bold if they match its theoretical value or it is in its theoretical estimated range. We also include, for each test graph, a measure of the {\em empirical degree of complexity} respect to the number of edges. This is defined as $\log(T(G))/\log(m)$, where $T(G)$ is the execution time of the algorithm in a graph $G$ with $m$ edges (measured in seconds). If the complexity is $\Theta(m^d)$, then this empirical degree converges to $d$ as $m$ goes to $\infty$. The empirical degree depends on the algorithm, the test graph, but also on the available computational power. When $m$ is not large enough, it should be understood as a local measure of complexity and it is used to compare the efficiency of different algorithms in the same graph. The (mean) empirical degree is placed as a subscript of the estimated congestions rounded to the first decimal digit.

The $\rm{sCD}$ algorithm depends on the choice of an initial spanning tree, this choice will be random. Moreover, some further steps depend on random choices. This is the reason because we decided to make $10$ executions of the program for each graph (with the exception of hypercube graphs). We give the minimum value attained by the sCD algorithms on these ten executions.

\subsection{Nonplanar graphs}
In this section, the sCD algorithms are tested on several complete, bipartite, multipartite, hypercube and random unweighted graphs. The results are summarized in the following Table:

\begin{table}[h]

\caption{Tests on some nonplanar graphs}
\centering
\begin{tabular}{|c|c|c|c|c|}
\hline
Graph	&	$\rm{sCD}_\infty$  & $\rm{sCD}_1$ & $\rm{sCD}_{10,\infty}$	\\ \hline
$K_{80}$	&	$\bm{79}_{1}$ & $\bm{79^2}_{1.0}$	& $\bm{79}_{1.1}$	\\ \hline
$K_{25,25}$	&	$\bm{48}_{0.9}$ & $\bm{1704}_{1.0}$ &	$\bm{48}_{0.9}$ \\ \hline
$K_{40,10}$	& $\bm{48}_{0.8}$	 & $\bm{1102}_{1.1}$ &	$\bm{48}_{0.9}$	\\ \hline
$K_{25,13,12}$	& $61_{0.9}$	 & $\bm{1920}_{1.0}$ &	$\bm{60}_{0.9}$	\\ \hline
$K_{25,12,12,1}$	& $61_{0.9}(\bm{38})$ & $\bm{1537}_{1.0}$ &	$\bm{38}_{1.1}$	\\ \hline
$K_{15,15,15,15,15}$	& $\bm{118}_{1.1}$ &  $\bm{5252}_{1.1}$ &	$\bm{118}_{1.1}$ 	\\ \hline
$H_7$	& $64_{1.1}$ &  $1828_{0.9}$ &	 $\bm{60}_{1.1}$	\\ \hline
$H_8$	& $128_{1.4}$ &  $4660_{1.0}$ &	$\bm{112}_{1.3}$	\\ \hline
$(\Z_{14})^2$	& $32_{1.3}$  &  $2010_{1.2}(1750)$ &	$30_{1.1}$	\\ \hline
$\Z_{20}\times\Z_{10}$	& $26_{1.2}$  &  $1622_{1.1}$ &	$22_{1.1}$	\\ \hline
$C_{6,6,6}$ & $\bm{36}_{1.4}$ &  $1860_{1.0}$ &	 ${6^2}_{1.1}$	\\ \hline
$C_{7,7,7}$ & $56_{1.1}$ &  $3326_{1.1}$ &	 ${7^2}_{1.3}$	\\ \hline
$C_{8,8,8}$ & $81_{1.2}$ &  $5432_{1.1}$ & ${8^2}_{1.4}$	\\ \hline
$\mathcal{G}(50,2\log(50)/50)$ & $[24,47]_{0.9}$  &  $[527,599]_{1.0}$ & $[23,38]_{1.0}$	\\ \hline
\end{tabular}

\end{table}

\begin{exmp}[Complete (multipartite) graphs]
On complete graphs both sCD algorithms reach the optimal tree with in all the tested examples. We conjecture that this will be always the case, this is equivalent to say that for any spanning tree there exists a congestion descent path to an optimal one.

On bipartite and multipartite graphs this is not always the case. Moreover, multipartite graphs of the form $K_{n_1,\dots,n_{k-1},1}$ are difficult to handle with $\rm{sCD}_\infty$ (since there is just one optimal tree), it is interesting to note that, in this case, the $\rm{sCD}_1$ algorithm provides better trees for the $L^\infty$-congestion (including the optimal one).
\end{exmp}

\begin{exmp}[Hypercube graphs]
This is one of the most interesting examples in this paper. In \cite{Hruska2008}, S.W. Hruska conjectured that $\cC_\infty(H_d)=2^{d-1}$, it was shown by Hiufai Law \cite{Hiufai2009} that $\cC_\infty(H_d) < 2^{d-1}$ for $d\geq 7$. The general estimations found in Table~\ref{tab:STC} are not very good for small values of $d$. Professor Hiufai Law kindly provided us (personal communication) a spanning tree with congestion $60$ for $H_7$ and it is known that $59$ is a lower bound. The $\rm{sCD}_\infty$ algorithm was not able to find congestions lower than $64$. However, the deeper $\rm{sCD}_{10,\infty}$ is able to reach spanning tree congestions in the range $[60,63]$ (with a frecuency of $1/50$ in the given implementation); this supports the conjecture that $\cC_\infty(H_7)=60$. We also provide the best spanning tree congestions obtained by $\rm{sCD}_{10,\infty}$ for $d=8,9,10$: $$96\leq \cC_\infty(H_8)\leq 112,\ 181\leq \cC_\infty(H_9)\leq 224,\ 341\leq \cC_\infty(H_{10})\leq 432\;.$$ In these dimensions the algorithm consistently gets spanning tree congestions\footnote{The algorithm was executed just once in $H_{10}$.} lower than $2^{d-1}$, and as far as we know, these results improve the already known best upper bounds given in \cite{Hiufai2009} for these dimensions. 

Respect to the $L^1$-congestions, we find a pattern with the help of \cite{OEIS} that, as far as we check gives the value of the $L^1$-congestion computed by $\rm{sCD}_1$. Based on this data, we conjecture that $\cC_1(H_d)=d\cdot (d-1)\cdot 2^{d-2}$ and it is obtained in the symmetric spanning tree obtained by doubling in an inductive way the unique spanning tree of $H_1$.  Observe also that these trees would be also solutions to the LSST problem for $H_n$, we do not know if this was observed before.

\end{exmp}

\begin{exmp}[Random graphs]
The $L^\infty$ congestion of random graphs $\mathcal{G}(n,p)$ were studied in \cite{Chandran-Cheung-Isaac2018} and \cite{Ostrovskii2011}. It is shown that, for $p>64\log(n)/n$, $\cC_\infty(\mathcal{G}(n,p))$ has linear growth w.h.p.. Our algorithm confirms that result and suggests that $\cC_\infty(\mathcal{G}(n,p))\leq n$ w.h.p. for $p=2\log(n)/n$.
\begin{figure}[h]
    \centering
\includegraphics[width = 10cm, height = 8cm]{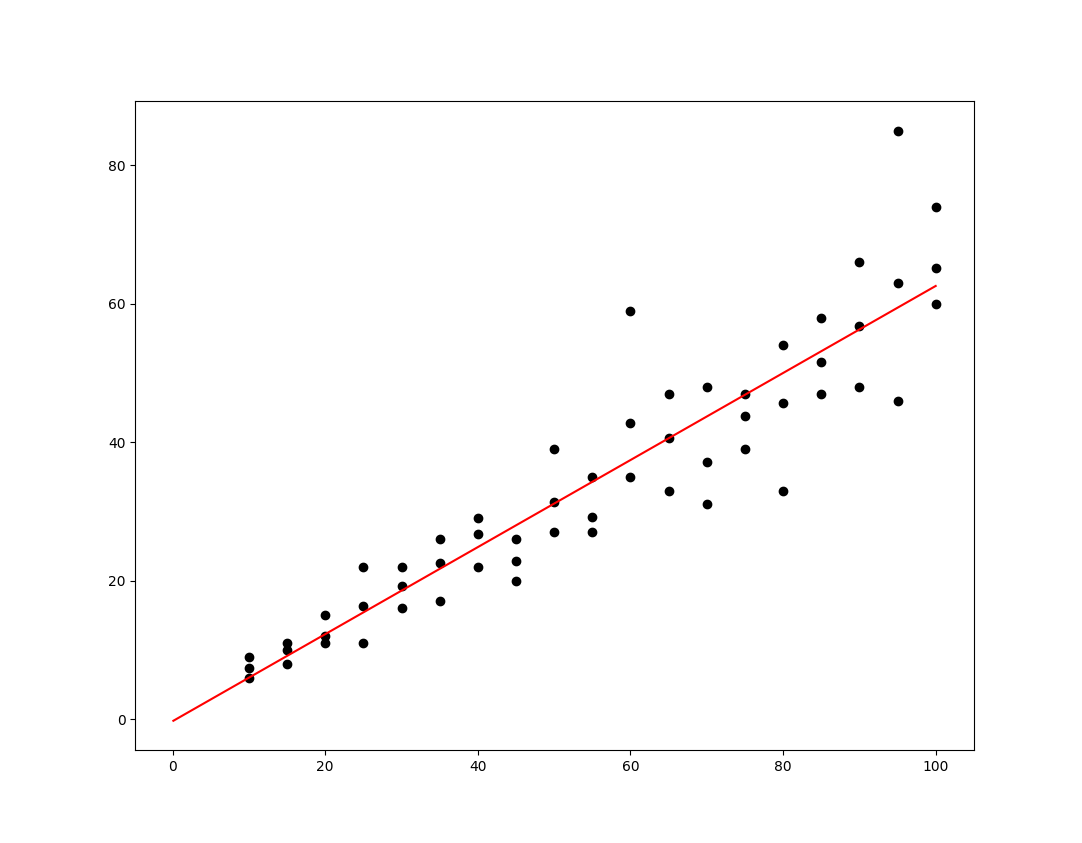}
    \caption{Some values of the estimated values of $\cC_\infty(G(n,p))$ for the $\rm{sCD}_{\infty}$ algorithm, for $n\in\{10,15,20,\dots,100\}$ and $p=2\log(n)/n$. Observe that the $L^\infty$-congestion has a linear trend as expected.}
    \label{fig:random_graphs}
\end{figure}

\end{exmp}

\subsection{Planar graphs}

In addition to the families of planar graphs introduced in Table~\ref{tab:STC}, we also include a family of random planar graphs called {\em planar uniform random graphs} and denoted by $PU_n$, where $n$ is the number of vertices.

The definition of a $PU_n$ graph is inductive:

\begin{itemize}
\item Suppose that a $PU_{k-1}$ graph has been defined for some $1\leq k \leq n$, assume that every vertex lie in $[0,1]\times[0,1]$ and the edges can be embedded in $\R^2$ as segment lines joining their adjacent vertices.

\item Choose a new point in $[0,1]\times[0,1]$ respect to the uniform distribution.

\item Connect this point to every vertex in $PU_{k-1}$ by an edge if and only if the correponding line segment does not meet any other edge in $PU_{k-1}$. This defines a graph $PU_k$.

\item Repeat inductively until $k=n$
\end{itemize}

The performance of our algorithms in the given test graphs is summarized in Table~\ref{table:test_planar_graphs}

\begin{table}[h]
\caption{Tests on some planar graphs}
\label{table:test_planar_graphs}
\centering
\begin{tabular}{|c|c|c|c|c|c|c|}
\hline
Graph	&	$\rm{sCD}_\infty$ &  $\rm{LOCBFS}_\infty$ & $\rm{ROC}_\infty$ & $\rm{sCD}_1$ & $\rm{LOCBFS}_1$ & $\rm{ROC}_1$ 	\\ \hline
$T_{20}$	& $\bm{26}_{1.1}$	& $\bm{26}_{0.6}$ & $\bm{26}_{0.5}$ & $1816_{1.1}$	& $2070_{0.6}$ & $2060_{0.5}$	\\ \hline
$T_{25}$	& $\bm{32}_{1.2}$	& $\bm{32}_{0.8}$ & $\bm{32}_{0.7}$ & $3126_{1.2}$	& $3768_{0.8}$ & $3716_{0.7}$	\\ \hline
$T_{30}$	& $\bm{40}_{1.1}$	& $\bm{40}_{0.9}$ & $\bm{40}_{0.9}$ & $4788_{1.2}$	& $5194_{1.0}$ & $5990_{0.9}$	\\ \hline
$R_{20,20}$	& $22_{1.2}$ & $\bm{20}_{0.5}$ & $\bm{20}_{0.5}$ & $2736_{1.1}$	& $3094_{0.6}$ & $3224_{0.6}$	\\ \hline
$R_{40,10}$	& $13_{1.0}$	& $\bm{11}_{0.5}$ & $\bm{11}_{0.5}$ & $2360_{1.0}$	& $2564_{0.6}$ & $2574_{0.6}$	\\ \hline
${\rm Hex}_{30}$	& $22_{1.1}(\bm{21})$	& $\bm{21}_{0.9}$ & $25_{0.8}$ & $5971_{1.0}$	& $6571_{0.9}$ & $6987_{0.7}$	\\ \hline
${\rm Hex}_{20,10}$	& $12_{0.9}$ & $\bm{11}_{0.4}$ & $\bm{11}_{1.1}$ & $2271_{0.9}$	& $2507_{0.6}$ & $2369_{0.5}$	\\ \hline
$PU_{100}$	& $[21,29]_{1.1}$	& $[19,28]_{0.2}$ & $[20,28]_{0.1}$ & $[551,593]_{1.2}$ 	& $[707,792]_{0.2}$ & $[746,847]_{0.1}$	\\ \hline
\end{tabular}

\end{table}

\subsection{Weighted graphs}

In this section we compare the algorithms in some weighted test graphs. Here we consider a very specific weight function. Given a vertex-weight function $\lambda:V\to\R$ we can consider the edge weights $\omega_+(uv)=\lambda(u)+\lambda(v)$ and $\omega_-(u,v)=|\lambda(u)-\lambda(v)|$, both edge-weights satisfy the triangle inequality and are interesting in practice. 

We consider these weight functions in some complete multipartite graphs and in euclidean grids. In our case, the vertex-weights are defined in a rather arbitrary way: just labelling the vertices from $0$ to $n-1$. For planar graphs we order vertices in a lexicographical way and for  complete multipartite graphs we order the partition sets by size. 

For complete multipartite graphs weighted by $\omega_{\pm}$, an upper bound is given by the spanning tree congestion on a maximal weighted tree among the minimum congestion spanning trees  in the respective unweighted graphs (these were described in subsection \ref{s:complete_graphs}). This seems to be the right choice for the $L^1$-STC. For complete, bipartite and multipartite with $n_k=1$, the following bounds hold:
\begin{align*}
&\cC_\infty(K^{\omega_-}_n)\leq \dfrac{n(n-1)}{2}\;,\quad \cC_1(K^{\omega_-}_n)\leq \dfrac{n(n-1)}{2} + \dfrac{n(n-1)(n-2)}{3}\;,\\
&\cC_\infty(K^{\omega_+}_n)\leq \dfrac{3(n-1)(n-2)}{2} + 1\;,\ \cC_1(K^{\omega_+}_n)\leq n^3 - 7n^2/2 + 9n/2 - 2\;,
\end{align*}

For $n\geq m\geq 2$:
\begin{align*}
&\cC_\infty(K^{\omega_-}_{n,m})\leq 2mn + \dfrac{m^2 + 2mn - 5m + n^2 - 7n}{2} + 4\;,\\ 
&\cC_1(K^{\omega_-}_{n,m})\leq \dfrac{3(n^2m+nm^2)}{2} - m^2 - n^2 - 4mn + 3n + 3m - 2\;,\\
&\cC_\infty(K^{\omega_+}_{n,m})\leq 3mn + (m^2  - 7m + 3n^2 - 13n)/2 + 6\;,\\ 
&\cC_1(K^{\omega_+}_{n,m})\leq \dfrac{9n^2m + 3nm^2}{2} - m^2 +  - 9mn + 5m - 3n^2 + 7n - 4\;.
\end{align*}

For $n_1\geq \cdots \geq n_k$ with $n_k=1$ and defining $n_0 = 0$, $N=n_1+\cdots+n_{k-1}+1$:

\begin{align*}
&\cC_\infty(K^{\omega_-}_{n_1,\dots,n_{k-1},1})\leq (N-2)(N-n_{k-1}-1) - \dfrac{(N-n_{k-1}-1)(N-n_{k-1}-2)}{2} + 1\\
&\cC_1(K^{\omega_-}_{n_1,\dots,n_{k-1},1})\leq \dfrac{(N-1)\cdot N}{2} + \sum_{i=1}^{k-2}n_i(N-n_0-n_1-\cdots -n_i-1)(N-n_0-n_1-\cdots -n_{i-1}-1)\\
&\cC_\infty(K^{\omega_+}_{n_1,\dots,n_{k-1},1})\leq (N-2)(N-n_{k-1}+1) + \dfrac{(N-n_{k-1}-1)(N-n_{k-1}-2)}{2} + 1\\
&\cC_1(K^{\omega_+}_{n_1,\dots,n_{k-1},1})\leq \dfrac{(N-1)\cdot(3N-4)}{2} + \\
&\sum_{i=1}^{k-2}\left(4(n_0 + n_1 \cdots n_{i-1}) + 3n_i + n_{i+1}+\cdots+n_{k-1}-1)\cdot n_i\cdot(n_{i+1}+\cdots + n_{k-1})\right)
\end{align*}

For the Euclidean grids, the set of minimal spanning trees to study is too large to analyse  its maximal weighted trees in a simple way.

In the $PU_n$ family of random planar graphs, the edge-weights are given by the euclidean distance, this is marked with a superscript $d$. There are no theoretical sharp bounds for these specific test graphs and this subsection is presented just to compare the performance of the algorithms.

In any case, we choose weights that satisfy the triangle inequality on edges, i.e., these are metric graphs. A further study of congestion in metric graphs will be considered in the future.

The performance of the algorithms in these weighted graphs are summarized in Tables~\ref{table:weighted_graphs} and \ref{table:weighted_planar_graphs}.

\begin{table}
\caption{Tests on some weighted nonplanar graphs}
\label{table:weighted_graphs}
\centering
\begin{tabular}{|c|c|c|c|c|c|}
\hline
Graph	&	$\rm{sCD}_\infty$  & $\rm{sCD}_1$ & $\rm{sCD}_{10,\infty}$ & $L^\infty$-Bound & $L^1$-Bound	\\ \hline
$K^{\omega_-}_{25}$	&	$300_{0.9}$ & $4900_{0.9}$ 	& $300_{0.8}$& $300$ & $4900$	\\ \hline
$K^{\omega_+}_{25}$	&	$852_{0.8}(829)$ & $13548_{0.9}$ 	& $852_{0.9}$&$829$ & $13548$	\\ \hline
$K^{\omega_-}_{15,15}$	&	$498_{0.9}$ & $8863_{1.0}$  & $498_{0.9}$&$589$ & $8863$\\ \hline
$K^{\omega_+}_{15,15}$	&	$890_{0.9}$ & $17501_{1.0}$  &  $838_{1.0}$&$981$ & $17501$\\ \hline
$K^{\omega_-}_{20,10}$	&	$429_{0.9}$ & $7788_{1.0}$  & $423_{0.9}$ &$759$ & $7788$\\ \hline
$K^{\omega_+}_{20,10}$	&	$985_{0.9}$ & $18086_{1.0}$  & $981_{0.9}$ &$1091$ & $18086$\\ \hline
$K^{\omega_-}_{15,10,6,1}$	& $496_{0.9}(451)$  & $8896_{1.0}$ &	$451_{0.9}$ & $451$	& $8896$ \\ \hline
$K^{\omega_+}_{15,10,6,1}$	& $1111_{1.0}$  & $21226_{1.0}$ &	$1111_{1.0}$&$1111$	& $21226$ \\ \hline
\end{tabular}
\end{table}

\begin{table}
\caption{Tests on some weighted planar graphs}
\label{table:weighted_planar_graphs}
\centering
\begin{tabular}{|c|c|c|c|c|c|c|}
\hline
Graph	&	$\rm{sCD}_\infty$ &  $\rm{LOCBFS}_\infty$ & $\rm{ROC}_\infty$ & $\rm{sCD}_1$ & $\rm{LOCBFS}_1$ & $\rm{ROC}_1$ 	\\ \hline
$T^{\omega_-}_{15}$	& $109_{1.0}$	& $116_{0.3}$ & $188_{0.2}(138)$ & $5556_{0.9}$	& $6030_{0.4}(6011)$ & $6436_{0.2}(6149)$	\\ \hline
$T^{\omega_+}_{15}$	& $2668_{1.0}$	& $2723_{0.3}$ & $4105_{0.2}(3921)$ & $99636_{0.9}$	& $104551_{0.3}$ & $122311_{0.2}$	\\ \hline
$R^{\omega_-}_{20,10}$	& $30_{0.9}$ & $121_{0.2}$ & $50_{0.2}$ & $4740_{0.9}$	& $7312_{0.3}(6868)$ & $6046_{0.2}$	\\ \hline
$R^{\omega_+}_{20,10}$	& $2386_{1.0}$ & $2728_{0.2}$ & $2708_{0.2}(2386)$ & $200112_{0.9}$	& $225567_{0.3}$ &$206784_{0.2}$ 	\\ \hline
$R^{\omega_-}_{15,15}$	& $29_{1.0}$ & $129_{0.2}$ & $47_{0.3}$ & $4928_{0.9}$	& $10230_{0.4}$ & $6496_{0.3}$	\\ \hline
$R^{\omega_+}_{15,15}$	& $3728_{1.1}$ & $3885_{0.2}(3763)$ & $5591_{0.3}(3477)$ & $286238_{0.9}$	& $307842_{0.4}$ & $294968_{0.3}$	\\ \hline
${\rm Hex}^{\omega_-}_{15}$	& $34_{0.9}$	& $89_{0.3}(59)$ & $88 _{0.3}(87)$ & $4160_{0.9}$	& $5190_{0.5}$ & $8320_{0.3}$	\\ \hline
${\rm Hex}^{\omega_+}_{15}$	& $3835_{1.0}$	& $3835_{0.2}$ & $3919 _{0.2}(3853)$ & $338472_{0.9}$	& $360605_{0.4}$ & $349676_{0.2}$	\\ \hline
$PU^{d}_{100}$	& $[4.9, 8.4]_{1.1}$	& $[6.4, 12.0]_{0.2}$ & $[6.7,11.7]_{0.1}$ & $[109.5, 126.0]_{1.0}$ 	& $[164.1, 205.5]_{0.2}$ & $[196.0, 271.7]_{0.2}$	\\ \hline
\end{tabular}

\end{table}

\subsection{Conclusions}

The family of $\rm{sCD}$ algorithms show a reasonable performance. Even for planar graphs they are, overall, much better than the LOCBFS and ROC algorithms.

LOCBFS and ROC algorithms are much faster than SCD and, for planar graphs, have a reasonable performance, but the $\rm{sCD}$ algorithm gives, overall, better results.

Recall that all the presented algorithms admit potential improvements that could lead to better performance and efficiency. This work is a first try showing that descent algorithms can be used to obtain reasonable bounds for congestion in mid range graphs (in the order of $10^3$ edges and vertices) and has the benefit to be already available to every interested reader (we do not know any other congestion algorithm accessible in a public repository).

A discussion about the empirical complexity of the $\rm{sCD}$ algorithms, that is rounding $1$, must be also done. This means that the time (in seconds) necessary to make an execution is in the order of the number of edges. This is much lower than our given estimation. There are several reasons for this difference. On the one hand, the computation of the edge congestions of a spanning tree $T$ has, typically, effective complexity of $m \diam(T)$. On the other hand, the number of iterations needed to reach a local minimum spanning tree is, in general, much lower than $mn$; the classical congestions of the test graphs is lower than both $n$ and $m$. Finally, in most of the presented test graphs, the algorithm finds edges that improve congestion very fast and compute the edge-congestions in a small cycle (that was roughly bounded by the number of vertices $n$). The latest iterations accumulate almost all the work load. This suggests a lower effective complexity for these test graphs. Hypercubes and cubic grids (where $n$ and $m$ are of similar order) have demonstrated to be the most challenging graphs from the point of view of computation, but still the empirical complexity of the $\rm{sCD}$ algorithm ranges in a subquadratic level. Nevertheless, we expect an increase in the empirical complexity as $m$ goes to infinity, converging to the effective degree of  complexity of the algorithm, for each family of test graphs.

Finally, observe that a general approximation theorem for the given algorithms is missing, i.e., we still do not know what is the asymptotic behavior of the algorithm respect to the  $L^p$-STC of a generic graph. This is left for a future work.

\end{document}